\documentclass[12pt, reqno]{amsart}

\makeatletter
\g@addto@macro{\endabstract}{\@setabstract}
\makeatother

\makeatletter
\def\@seccntformat#1{%
	\ifstrequal{#1}{subsection}% determine if it is subsection
	{\bfseries\csname the#1\endcsname.}% subsection → boldface
	{\csname the#1\endcsname.} % else → preserve the default format
}
\makeatother

\usepackage{graphics, stackrel}
\usepackage{amsmath, amssymb, amsthm}
\usepackage{graphicx}
\usepackage{verbatim}
\usepackage{amsfonts}
\usepackage{filecontents}
\usepackage[round,authoryear]{natbib}
\usepackage{fancyvrb}
\usepackage{color}

\usepackage[citecolor=blue, colorlinks=true, linkcolor=blue]{hyperref}

\usepackage{mathrsfs}
\usepackage{bbm}

\usepackage[left=1.25in, right=1.25in, top=1.0in, bottom=1.15in, includehead, includefoot]{geometry}

\usepackage{multirow}
\usepackage{algorithm}
\usepackage{enumitem}
\usepackage{booktabs}
\usepackage{diagbox}
\usepackage{makecell}
\usepackage{rotating}
\usepackage{adjustbox}
\usepackage[graphicx]{realboxes}
\usepackage{lscape}
\usepackage{hhline}
\newcolumntype{C}[1]{>{\centering\arraybackslash}p{#1}}
\usepackage{float}

\newcommand{\iidsim}{\stackrel {\textrm{ {\sc iid }}} {\sim} }
\newcommand{\1}{\mathbbm 1}

\newcommand*\diff{\mathop{}\!\mathrm{d}}
\renewcommand{\epsilon}{\varepsilon}
\renewcommand{\phi}{\varphi}

\newcommand{\cC}{\mathscr C}

\newcommand{\hH}{\mathscr H}

\newcommand{\fF}{\mathscr F}

\newcommand{\ZZ}{\mathsf Z}
\renewcommand{\SS}{\mathsf S}
\newcommand{\TT}{\mathsf \Theta}

\newcommand{\pP}{\mathscr P}

\newcommand{\RR}{\mathbbm R}

\newcommand{\NN}{\mathbbm N}
\newcommand{\PP}{\mathbbm P}

\newcommand{\EE}{\mathbbm E}

\theoremstyle{plain}

\newtheorem{theorem}{Theorem}[section]

\newtheorem{lemma}{Lemma}[section]
\newtheorem{proposition}{Proposition}[section]

\theoremstyle{definition}

\newtheorem{assumption}{Assumption}[section]

\DeclareMathOperator{\diag}{diag}
\renewcommand{\underline}[1]{\text{\b{$#1$}}} % underbar similar to \bar{}

\usepackage{caption}

\usepackage[flushleft]{threeparttable}

\newcommand{\me}{\mathrm{e}}

\usepackage{ dsfont }

\newcommand{\vertiii}[1]{{\left\vert\kern-0.25ex\left\vert\kern-0.25ex\left\vert #1 
		\right\vert\kern-0.25ex\right\vert\kern-0.25ex\right\vert}}

\usepackage{subcaption}
\usepackage{graphicx}  

\makeatletter
\newcommand*{\rom}[1]{\expandafter\@slowromancap\romannumeral #1@}
\makeatother

\usepackage{stackengine}

\usepackage{manyfoot}
\DeclareNewFootnote{default}  
\DeclareNewFootnote{symbol}[fnsymbol]

\makeatletter
\renewcommand{\@makefnmark}{%
	\hbox{\@textsuperscript{\normalfont\footnotesize\@thefnmark}}%
}
\makeatother

\begin{document}
	
\title{}

\date{\today}

\begin{center}
	\LARGE 
	%Optimal Savings under Uncertainty and Learning about Economic Transitions
	Optimal Savings under Transition Uncertainty and Learning Dynamics\footnotesymbol{We thank Weiming Li, Nisha Peng, Xinxi Song, John Stachurski, Bin Wu, Shenghao Zhu, and seminar audiences at CUEB for helpful comments and suggestions. Qingyin Ma gratefully acknowledges financial support from the National Natural Science Foundation of China (Grant No. 72573111).}

	\vspace{1em}
	\normalsize
	Qingyin Ma\footnotesymbol{
		%International School of Economics and Management,  
		Capital University of Economics and Business. Email address: \texttt{qingyin.ma@cueb.edu.cn}.}
	\quad \text{and} \quad 
	Xinxin Zhang\footnotesymbol{Beijing University of Posts and Telecommunications. Email address: \texttt{xinxin.zhang@bupt.edu.cn}}
	\par 
	
	\vspace{1em}
	\normalsize{\today}
\end{center}

\begin{abstract} 
	This paper studies optimal consumption and saving decisions under uncertainty about the transition dynamics of the economic environment. We consider a general optimal savings problem in which the exogenous state governing discounting, capital returns, and nonfinancial income follows a Markov process with unknown transition probability, and agents update their beliefs over time through Bayesian learning. Despite the added endogenous state from belief updating, we establish the existence, uniqueness, and key structural properties of the optimal policy, including monotonicity and concavity. We also develop an efficient computational method and use it to study how transition uncertainty and learning interact with precautionary motives and wealth accumulation, highlighting a dynamic mechanism through which uncertainty about regime persistence shapes consumption dynamics and long-run household wealth.
%	This paper studies optimal consumption and saving decisions when households face uncertainty about the transition dynamics of the economic environment. We consider a general optimal savings problem in which the exogenous state governing discounting, capital returns, and nonfinancial income follows a Markov process with unknown transition probability, and agents update their beliefs over time through Bayesian learning. Despite the additional endogenous state introduced by belief updating, we establish the existence and uniqueness of the optimal policy and characterize its key structural properties, including monotonicity and concavity. These results extend the theoretical foundations of optimal savings models to environments with transition uncertainty. We also develop an efficient computational method and use it to study how transition uncertainty and learning interact with precautionary motives and wealth accumulation, highlighting a dynamic mechanism through which uncertainty about regime persistence shapes consumption dynamics and long-run household wealth.

    \vspace{1em}
    
    % C61: Optimization Techniques; Programming Models; Dynamic Analysis
    % D14: Household Saving; Personal Finance
    % D15: Intertemporal Household Choice; Life Cycle Models and Saving
    % D52: Incomplete Markets
    % E21: Consumption; Saving; Wealth
    % G51: Household Saving, Borrowing, Debt, and Wealth
\end{abstract}

\section{Introduction}

Households routinely make economic decisions in environments where the future evolution of the economy is difficult to assess. While much of the literature focuses on uncertainty about the realization of shocks, in many settings agents also face uncertainty about the transition dynamics governing the evolution of economic states. For example, households may be uncertain about how persistent business cycle conditions are, how long favorable financial conditions will last, or how likely policy regimes are to change \citep{sargent1999conquest, hansen2001robust, evans2001learning, baker2016measuring}. In such environments, agents do not know the true transition law governing the underlying state process and must learn about it over time as new observations arrive. We refer to this form of uncertainty as \emph{transition uncertainty}. Understanding how such uncertainty shapes household consumption and saving decisions is therefore an important question for both macroeconomic theory and quantitative analysis.

A large body of research studies optimal consumption and saving under income, preference, or return risk. Canonical models of the income fluctuation problem and buffer-stock saving emphasize the role of precautionary motives in shaping household behavior. In these frameworks, households smooth consumption in response to uninsurable income fluctuations, preference shocks, and asset return risk. However, most of this literature assumes that the data generating process governing the economic environment is known to agents. In particular, the transition probabilities of exogenous state variables are typically treated as common knowledge \citep{li2014solving, benhabib2015wealth, ma2020income, ma2026theory}.

In reality, households rarely know the true transition structure of the economy. Instead, they must form beliefs about it based on limited observations and update these beliefs as new data arrive. A growing literature studies learning about macroeconomic regimes and model uncertainty. Much of this work focuses on asset pricing, macroeconomic policy, or aggregate dynamics, and is largely quantitative in nature \citep{evans2001learning, cogley2005drifts, pastor2013political, collin2016parameter}. By contrast, the implications of transition uncertainty and belief updating for optimal household consumption and saving decisions have received relatively little attention. In particular, there is limited work characterizing the existence, uniqueness, and structural properties of optimal policies in environments where agents must learn about the transition dynamics of the economic state.

This paper develops a general theory of optimal savings under transition uncertainty and learning. In particular, we consider a generic optimal savings framework in which the exogenous state governing discounting, capital returns, and nonfinancial income evolves according to a Markov process with unknown transition probability. Agents believe that the true transition kernel belongs to a set of candidate models and updates these beliefs according to Bayes' rule as new observations of the exogenous state arrive. Consequently, the agent's posterior belief about the transition dynamics becomes an additional endogenous state variable in the dynamic optimization problem.

Despite the added complexity of transition uncertainty and belief updating, we show that the optimal consumption policy remains well defined. Under mild conditions, we establish the existence and uniqueness of the optimal policy and characterize its key properties, including continuity, monotonicity, concavity, and asymptotic linearity. Our key assumption imposes a uniform stability condition on the expected discounted return process, requiring that the \emph{subjective} long-run growth rate of discounted returns remains bounded, preventing explosive accumulation incentives. Importantly, this condition has a transparent economic interpretation and can be easily verified in applications  since it depends only on observable primitives. %Taken together, these results extend the theoretical foundations of optimal savings problems to environments in which agents face transition uncertainty, such as uncertainty about the persistence of macroeconomic regimes.

Building on this theoretical framework, we develop an efficient computational method. The algorithm updates beliefs on a \emph{barycentric coordinate grid} and extends the endogenous grid point method of \cite{carroll2006method} to environments in which beliefs evolve through Bayesian learning. This approach allows us to compute the optimal policy in settings where both economic states and endogenous beliefs jointly determine household decisions, while avoiding the costly root-finding procedures typically required in dynamic programming problems.

We then apply the theory and computational method to study the quantitative implications of transition uncertainty. The application illustrates how learning about regime persistence interacts with precautionary behavior and wealth accumulation. When agents initially assign positive probability to less persistent regimes, perceived future risks are amplified, strengthening precautionary saving and leading households to accumulate larger buffer-stock wealth. As beliefs gradually converge toward the true transition dynamics, precautionary motives weaken, but the additional wealth accumulated during the learning phase persists, generating long-run differences in wealth holdings and consumption dynamics between economies with learning and those with full information. While belief updating introduces short-run fluctuations in perceived risks, the larger wealth buffers generated by precautionary saving ultimately improve households’ ability to smooth consumption over time.

Taken together, the paper provides a unified framework for analyzing optimal consumption and saving decisions under transition uncertainty and Bayesian learning. Our main contributions are twofold. First, we rigorously establish the existence, uniqueness, and key properties of the optimal consumption policy in environments where the transition structure is unknown, thereby extending the theoretical foundations of the income fluctuation problem. Second, we develop computational tools that render the model tractable and use them to illustrate the economic mechanisms through which transition uncertainty shapes precautionary behavior, consumption dynamics, and long-run household wealth.

More broadly, our analysis highlights a dynamic interaction between learning, precautionary motives, and wealth accumulation. Transition uncertainty initially amplifies precautionary saving, depressing consumption, but the accumulated wealth ultimately enhances long-run consumption smoothing. In this sense, uncertainty about economic regimes can simultaneously generate short-run volatility and long-run resilience in household balance sheets.

\subsection*{Related Literature}\label{ss:literature}

Our paper contributes to the large literature on optimal consumption and saving. Canonical studies of the income fluctuation problem examine how households smooth consumption under uninsurable income risk and borrowing constraints \citep{schechtman1976income, carroll1997buffer, carroll2001theory, kuhn2013recursive, li2014solving}, while more recent studies establish the existence, uniqueness, and structural properties of optimal policies in more general stochastic environments with preference and return risks \citep{chamberlain2000optimal, benhabib2015wealth, ma2020income, ma2021theory, ma2026theory}. Unlike these studies, we introduce uncertainty about the transition dynamics and allow agents to learn over time, showing that optimal savings decisions remain well behaved and that the policy can be characterized rigorously.

Our work also relates to research on model uncertainty, ambiguity, and robust control \citep{hansen2001robust, cogley2008robustness, epstein2014ambiguous}, where agents account for possible misspecification of the transition law. These studies typically assume ambiguity aversion rather than learning. By contrast, we focus on Bayesian updating of beliefs about the transition kernel, allowing us to study how learning interacts with precautionary saving within a general optimal savings framework.

Finally, we connect to the literature on Bayesian learning about macroeconomic regimes \citep{evans2001learning, cogley2005drifts, baker2016measuring}, which mainly focuses on asset pricing, policy, or aggregate dynamics. Our analysis differs by emphasizing optimal household consumption and saving, and by providing a general theoretical framework establishing existence, uniqueness, and structural properties of optimal policies under transition uncertainty and Bayesian learning.

The rest of this paper is structured as follows. Section~\ref{s:or} formulates the model and establishes a general theory on the existence, uniqueness, and computability of the optimal policy. Section \ref{s:ifp_properties} studies the structural properties of the optimal consumption policy. Section \ref{s:quant} develops the computational method and presents quantitative analysis of the economic mechanisms generated by transition uncertainty and learning. Proofs are deferred to the appendices.

\section{Optimality results}
\label{s:or}

In this section, we formulate a generic optimal savings problem under the framework of learning about the law of motion of exogenous states. We then discuss the existence, uniqueness, and computability of the optimal consumption policy.

\subsection{Problem Statement}
\label{ss:ifp_problem}

Time is discrete and indexed by $t=0,1,2,\dots$. Let $w_t$ denote the agent's financial wealth at the beginning of period $t$. Let $c_t, \beta_t, R_t$, and $Y_t$ denote consumption, the discount factor, the gross rate of return on wealth, and nonfinancial income in period $t$, respectively. We normalize the initial discount factor $\beta_0=1$. The agent's preferences are represented by a period utility function $u$. Let $Z_t$ be an exogenous state variable capturing idiosyncratic or aggregator factors that jointly shape agents' intertemporal tradeoffs, consumption and saving behavior. Let $\theta_t$ denote the agent's belief about the structure of the transition probabilities governing $\{Z_t\}$.

Given an initial condition $(w_0,Z_0,\theta_0)=(w,z,\theta)$, the agent chooses a consumption-wealth path $\{(c_t, w_t)\}$ to:
\begin{align}
	& \text{maximize}  
	&& \EE_0 \left\{ 
		\sum_{t = 0}^\infty 
		\left(\prod_{i=0}^t \beta_i \right) 
		u(c_t)
	\right\}    \label{eq:value} \\
	&\text{subject to}   
	&& w_{t+1} = R_{t+1} (w_t - c_t) + Y_{t+1} 
	\quad \text{and} \quad 0 \leq c_t \leq w_t.
	\label{eq:trans_at}
\end{align}
Here the discounting, return, and income processes evolve according to
\begin{equation}
	\label{eq:law_exog}
	\beta_t = \beta \left(Z_t, \epsilon_t \right), 
	\quad R_{t} = R \left(Z_{t}, \epsilon_t \right)
	\quad \text{and} \quad
	Y_{t} = Y \left(Z_{t}, \epsilon_t \right),
\end{equation}
where $\beta$, $R$, and $Y$ are nonnegative measurable functions, and $\{\epsilon_t\}$ is an {\sc iid} sequence with possibly continuous support. 
Both $Z_t$ and $\epsilon_t$ may be vector-valued.

Throughout, we impose the following regularity condition on preferences.

\begin{assumption}
	\label{a:utility}
	The utility function $u:[0,\infty)\to \RR \cup \{-\infty\}$ is twice differentiable on $(0,\infty)$ with $u'>0$, $u''<0$, and $u'(c) < 1$ as $c\to\infty$.
\end{assumption}

\subsection{Unknown Transition Probability and Learning}
\label{ss:bayes}

Suppose $\{Z_t\}$ is a time homogeneous Markov chain taking values in a finite set $\ZZ$, but the transition probability is unknown to the agent. The agent believes that the true transition matrix belongs to a finite set
\begin{equation*}
	\pP := \{P_1, \dots, P_N\}.
\end{equation*}
For each $i$ in $\{1, \dots, N\}$ and each pair of states $(z,\hat z)$ in $\ZZ\times\ZZ$, let $P_i(z,\hat z)$ denote the probability that $Z_{t+1}=\hat z$ conditional on $Z_t=z$ when the transition matrix is $P_i$. Let $\theta_t:= (\theta_{t1}, \dots, \theta_{tN})$ denote the agent's belief over the set of candidate transition matrices at time $t$, where $\theta_{ti}$ is the posterior probability assigned to $P_i$. The belief vector $\theta_t$ lies in the probability simplex 
\begin{equation*}
	\Theta := \left\{
  \theta \in \RR_+^N: \sum_{i=1}^{N} \theta_i = 1
	\right\}.
\end{equation*}
Given belief $\theta_{t}$, the agent's \emph{subjective transition probability} is the mixture
\begin{equation*}
	P_{\theta_t} (z,\hat z) := 
	\sum_{i=1}^{N} \theta_{ti} P_i(z,\hat z).
\end{equation*}
After observing a realization $(Z_t,Z_{t+1}) = (z,\hat z)$, beliefs are updated via Bayes' rule:
\begin{equation}\label{eq:bayes}
	\theta_{t+1,i} 
	= \frac{P_i(z,\hat z) \theta_{ti}}{
		\sum_{j=1}^{N} P_j(z,\hat z) \theta_{tj}},
	\qquad i = 1, \dots, N.
\end{equation}
Hence, the belief process $\theta_t$ evolves endogenously as a function of the observed state transitions and forms part of the agent’s state vector.

We pair $\ZZ$ with a partial order $\preceq$ and assume that $(\ZZ,\preceq)$ is a complete partially ordered set with elements $\{z_1, \dots, z_M\}$. A function $h:\ZZ \to \RR$ is called \emph{nondecreasing} or \emph{increasing} if $h(z) \leq h(\hat z)$ for all $z,\hat z\in \ZZ$ with $z\preceq \hat z$. 

With slight abuse of notation, we also use $\preceq$ to denote first order stochastic dominance over transition matrices on $\ZZ$. A transition matrix $P$ on $\ZZ$ is said to be \emph{monotone} if for all increasing function $h:\ZZ \to \RR$, the function $Ph(z) := \sum_{\hat z\in \ZZ} P(z,\hat z) h(\hat z)$ is increasing in $z$.

\subsection{Optimality: Definitions and Fundamental Properties}
\label{ss:opt}

Let $\SS$ be the state space of the state process $\{(w_t,Z_t,\theta_t)\}_{t\geq 0}$. A \emph{feasible policy} is a Borel measurable function $c \colon \SS \to \RR$ such that $0 \leq c(w,z,\theta) \leq w$ for all $(w,z,\theta) \in \SS$. Given a feasible policy $c$ and an initial condition $(w,z,\theta)$, the induced wealth process evolves according to the law of motion \eqref{eq:trans_at}, with consumption at time $t$ determined by $c_t = c (w_t, Z_t, \theta_t)$. 

For clarity, throughout the paper, expectations conditional on the initial state are denoted by $\EE_{w, z, \theta} [\,\cdot\,]$ when conditioning on $(w_0, Z_0, \theta_0) = (w, z,\theta)$, and by $\EE_{z, \theta}[\,\cdot\,]$ when conditioning on $(Z_0,\theta_0) = (z,\theta)$. Additionally, for any random variable $X$, its next period value is denoted by $\hat X$. A more detailed explanation of these notations can be found in the Appendix.

To define optimality, we adopt the overtaking criterion of \cite{brock1970axiomatic}. For a feasible policy $c$, we consider the expected sum of utilities up to time $L$, denoted by
\begin{equation*}
	V_{c,L} (w,z, \theta) := \EE_{w,z,\theta}\sum_{t=0}^{L} 
	\left(\prod_{i=0}^t \beta_{i} \right)
	(u \circ c) (w_{t},Z_{t},\theta_t).
\end{equation*}
Given two feasible policies $c_1$ and $c_2$, we say that $c_1$ \textit{overtakes} $c_2$ if 
\begin{equation*}
	\limsup_{L\to\infty} \left[
	    V_{c_2,L} (w,z,\theta) - V_{c_1,L} (w,z,\theta)
	\right] \leq 0
\end{equation*}
for all initial state $(w,z,\theta)\in \SS$. A feasible policy $c^*$ is 
\textit{optimal} if it overtakes any other feasible policy $c$. A feasible 
policy $c$ is said to satisfy the \textit{first order optimality condition} if
\begin{equation}\label{eq:foc}
	(u'\circ c)(w,z, \theta) = \min \left\{
		\max \left\{
			\EE_{z,\theta} \hat \beta \hat R 
			(u'\circ c)(\hat w,\hat Z,\hat \theta), 
			u'(w)
		\right\}, u'(0)
	\right\}
\end{equation}
for all $(w,z, \theta) \in \SS$, where the next period wealth satisfies
\begin{equation}\label{eq:w_hat}
	\hat w = \hat R \left[w-c(w,z,\theta)\right] + \hat Y.
\end{equation}
The first order optimality condition represents the optimal tradeoff between consumption today and saving for future periods. Under normal circumstances, the marginal utilities of consumption and saving are balanced. However, corner solutions can arise. When wealth is low, the agent may consume all available wealth ($c=w$) to maximize immediate utility. Conversely, when $u'(0)<\infty$, the agent might save everything, consuming nothing today ($c=0$) to accumulate wealth for a higher future return.\footnote{In general, the 
	minimization over $u'(0)$ is necessary due to the stochastic setup. In cases of large and highly persistent discounting or return processes, where $\EE_{z,\theta} \hat \beta \hat R > 1$, the marginal reward of saving, $\EE_{z,\theta} \hat \beta \hat R (u'\circ c)(\hat w,\hat Z,\hat\theta)$, can exceed $u'(0)$, requiring the minimization to satisfy the equilibrium condition $u'(c)\leq u'(0)$. When $u'(0)=\infty$, the minimization operation no longer applies.}

A feasible policy $c$ is said to satisfy the \textit{transversality condition} if, for all $(w,z, \theta) \in \SS$, 
\begin{equation}\label{eq:tvc}
	\lim_{t \to \infty} \EE_{w,z,\theta} 
	    \left(\prod_{i=0}^t \beta_{i}\right)
	(u'\circ c)(w_t, Z_t, \theta_t) w_t=0.
\end{equation}
The transversality condition rules out excessive long run wealth accumulation, preventing policies that defer consumption indefinitely.

The following result demonstrates that the first order and transiversality 
conditions are sufficient for optimality. The proof is similar to 
Proposition~15.2 and Theorem~15.3 of \cite{toda2024essential} and thus omitted.

\begin{proposition}[Sufficiency of first order and transversality conditions]
	\label{pr:suff_euler_tvc}
	If Assumption~\ref{a:utility} holds, then every feasible policy satisfying 
	the first order and transversality conditions is an optimal policy.
\end{proposition}

\subsection{Optimal Policy}
\label{ss:exis_uniq}

We now study the existence, uniqueness, and computability of a feasible policy satisfying the first order condition~\eqref{eq:foc}. Recall the stochastic setup defined in~\eqref{eq:law_exog}. We make the following regularity assumption.

\begin{assumption}
	\label{a:finite_exp}
	The following conditions are true:
	\begin{enumerate}
		\item\label{i:pos_inc_bR} $\beta(z,\epsilon)$ and $R(z,\epsilon)$ are strictly positive and nondecreasing in $z$.
		\item\label{i:fin_exp} $\EE_z \beta Y < \infty$, $\EE_z \beta u'(Y) < \infty$, and $\EE_{z} \beta R u'(Y) < \infty$ for all $z\in\ZZ$.
	\end{enumerate}
\end{assumption}

To study optimality, we define the state space as
\begin{equation*}
	\SS := \begin{cases}
		(0, \infty) \times \ZZ \times \TT, & \text{if } u'(0) = \infty, \\
		[0, \infty) \times \ZZ \times \TT, & \text{otherwise}. 
	\end{cases}
\end{equation*}
In particular, when $u'(0) = \infty$, Assumption~\ref{a:finite_exp} ensures that nonfinancial income is strictly positive almost surely for all $t \geq 1$. As a result, wealth remains strictly positive with probability one at all future dates, so excluding zero from the wealth space does not affect optimality. 

Let $\cC$ be the set of continuous functions $c: \SS \to \RR_+$ such that $w \mapsto c(w,z, \theta)$ is increasing for all $(z,\theta)\in \ZZ\times \TT$, with the constraints $0\leq c(w,z, \theta) \leq w$ for all $(w,z, \theta)\in \SS$, and 
\begin{equation}\label{eq:upc_bd}
	\sup_{(w,z,\theta)\in \SS} \left|
	(u'\circ c)(w,z, \theta) - u'(w)
	\right| < \infty.
\end{equation}
We equip $\cC$ with a metric $\rho$ defined as follows: for each $c_1, c_2 \in \cC$, 
\begin{equation*}
	\rho(c_1,c_2) := \sup_{(w,z, \theta)\in \SS} \left|
	    (u'\circ c_1)(w,z, \theta) - (u'\circ c_2)(w,z, \theta)
	\right|,
\end{equation*}
which evaluates the maximal difference in terms of marginal utility. While
elements of $\cC$ are not generally bounded, the metric $\rho$ is well defined. In particular, $\rho$ is finite on $\cC$ since the 
triangular inequality implies that
$\rho(c_1,c_2) \leq \left\| u' \circ c_1 - u' \right\| 
+ \left\| u' \circ c_2 - u' \right\|$, where $\|\cdot\|$ is the standard 
supremum norm, and the last two terms are finite by~\eqref{eq:upc_bd}. 
$(\cC, \rho)$ is a complete metric space. The proof is a simple extension of Proposition~4.1 of \cite{li2014solving} and thus omitted.

We characterize the optimal policy as the fixed point of a \emph{time iteration operator} $T$. For a policy $c \in \cC$ and a state $(w,z, \theta) \in \SS$, the image of $c$ under $T$, denoted by $Tc(w,z, \theta)$, is the unique $\xi \in (0,w]$ that solves
\begin{equation}
	\label{eq:tio}
	u'(\xi) = \min\left\{
		\max \left\{
			\EE_{z,\theta} \hat \beta \hat R 
			(u'\circ c)\big(\hat R (w-\xi) + \hat Y, \hat Z, \hat\theta\big),
			u'(w)
		\right\}, u'(0)
	\right\}.
\end{equation}
To ensure that $T$ is a well defined self-map on $\cC$ and admits a unique fixed point, we impose an additional assumption. 
For each candidate transition matrix $P_i$ and $\alpha\in \{0,1\}$, we define the matrix\footnote{By 
	definition, $\EE_z \beta R^\alpha = \EE_z \beta(z,\epsilon) R(z,\epsilon)^\alpha$ for $z\in \ZZ$, where the expectation is taken with respect to the {\sc iid} innovation $\epsilon$ holding $z$ fixed.}
\begin{equation}\label{eq:D_mat}
	D_\alpha :=\diag \{\EE_{z_1} \beta R^\alpha, \cdots, \EE_{z_M} \beta R^\alpha \}.
\end{equation}
Hence, $D_\alpha$ is the diagonal matrix whose entries are the expected discounted returns. Moreover, recall that the spectral radius $r(A)$ of a square matrix $A$ is defined as the maximum absolute value of all its eigenvalues, that is,
\begin{equation*}
	r(A) := \max\{|\lambda|: \lambda \text{ is an eigenvalue of } A\}.
\end{equation*}

\begin{assumption}\label{a:unif_sr}
	There exists an irreducible and monotone element $P^*\in \pP$ such that $P_i \preceq P^*$ for all $i=1,\dots,N$. Moreover, $r(P^*D_\alpha) < 1$ for $\alpha\in \{0,1\}$.
\end{assumption}

Assumption~\ref{a:unif_sr} imposes a uniform stability restriction on discounted expected returns across all candidate transition matrices. The spectral radius condition $r(P^* D_\alpha) < 1$ ensures that, under the most favorable perceived transition dynamics, the asymptotic growth rate of discounted expected returns remains sufficiently small to rule out explosive saving behavior. This condition does not preclude the possibility that $\beta_t \geq 1$ or $\beta_t R_t \geq 1$ in any given period. Rather, it substantially generalizes the standard assumption $\beta<1$ and $\beta R < 1$ in the classical optimal savings problem, where both $R_t \equiv R$ and $\beta_t\equiv \beta$ are constant.

When $N=1$, there is no learning and the agent knows the true transition probability $P$. In this case, Assumption~\ref{a:unif_sr} reduces to $r(P D_\alpha)<1$, which generalizes the key spectral radius assumption in \cite{ma2020income}.

The next proposition ensures that any element in the candidate space that satisfies the first order condition is indeed an optimal policy. 

\begin{proposition}[Sufficiency of first order condition]
	\label{pr:suff_foc}
	Let Assumptions~\ref{a:utility}--\ref{a:finite_exp} hold. If $c \in \cC$ and satisfies the first order condition~\eqref{eq:foc}, then $c$ satisfies the transversality condition~\eqref{eq:tvc}. In particular, $c$ is an optimal policy.
\end{proposition}

The following theorem establishes the existence and uniqueness of a candidate 
policy that satisfies the first order optimality condition.

\begin{theorem}[Existence, uniqueness, and computability of optimal policies]
	\label{t:opt}
	If Assumptions~\ref{a:utility}--\ref{a:finite_exp} hold, then 
	the following statements are true:
	\begin{enumerate}
		\item $T: \cC \to \cC$ is well defined and has a unique fixed point $c^{*} \in \cC$.
		\item The fixed point $c^{*}$ is the unique optimality policy in $\cC$.
		\item\label{tr:conv} For all $c \in \cC$, we have 
		$\rho(T^{k}c,c^{*}) \to 0$ as $k \to \infty$.
	\end{enumerate}
\end{theorem}

This convergence result demonstrates that the time iteration algorithm converges globally within the candidate space $\cC$.

\section{Properties of Consumption and Saving}
\label{s:ifp_properties}

We next examine the key properties of the optimal consumption function
characterized in Theorem~\ref{t:opt}, maintaining Assumptions~\ref{a:utility}--\ref{a:unif_sr} throughout. We begin with monotonicity results.

\begin{proposition}[Monotonicity with respect to wealth]
	\label{pr:monotonea}
	The optimal consumption and savings functions $c^*(w,z,\theta)$ and 
	$i^*(w,z,\theta) := w - c^*(w,z,\theta)$ are increasing in $w$.
\end{proposition}

\begin{proposition}[Monotonicity with respect to income]
	\label{pr:monotoneY} 
	If $\{ Y_{1t} \}$ and $\{ Y_{2t} \}$ are two income processes satisfying 
	$Y_{1t}\leq Y_{2t}$ for all $t$ and $c_1^*$ and $c_2^*$ are the
	corresponding optimal consumption functions, then $c_1^* \leq c_2^*$
	pointwise on $\SS$.
\end{proposition}

These results formalize the intuitive idea that households consume and save more when they are wealthier or receive higher income, reflecting the monotone response of optimal decisions to favorable financial conditions.

Next, we characterize the threshold below which the borrowing constraint binds.

\begin{proposition}[Threshold for saving decision]
	\label{pr:binding}
	For all $c \in \cC$, there exists a threshold $\bar{w}_c(z,\theta)$ such that $Tc(w,z,\theta) = w$ if and only if $w \leq \bar{w}_c (z,\theta)$. In particular, letting 
	\begin{equation*}
		\bar w(z,\theta) := (u')^{-1} \left[ \min\left\{
			    \EE_{z,\theta} \hat \beta \hat R 
			    (u'\circ c^*)(\hat Y,\hat Z,\hat \theta),
			    u'(0)
		    \right\}
		\right],
	\end{equation*}
	we have $c^*(w,z,\theta) = w$ if and only if $w \leq \bar{w}(z,\theta)$.
\end{proposition}

This result formalizes the intuitive concept of a wealth threshold. Households are constrained and consume all available resources if and only if wealth falls below $\bar w(z,\theta)$.

We next establish a uniform lower bound on consumption.

\begin{proposition}[Lower bound on consumption]
	\label{pr:lb}
	If there exists a constant $\bar s \in (0,1)$ such that 
	\begin{equation}\label{eq:sbar}
		\EE_{z,\theta} \hat \beta \hat R u'(\bar s \hat R w) \leq u'(w) 
		\quad \text{for all } w>0 \text{ and } (z,\theta)\in \ZZ\times \TT,
	\end{equation}
	then $c^*(w,z,\theta) \geq (1-\bar s)w$ for all $(w,z,\theta)\in \SS$.
\end{proposition}

Intuitively, this condition ensures that households always consume a positive fraction of wealth, preventing extreme precautionary saving that would drive consumption arbitrarily close to zero.

Finally, we turn to concavity and asymptotic properties under constant relative risk aversion (CRRA) preferences:
\begin{equation}\label{eq:crra}
	u(c) = \begin{cases}
		\frac{c^{1-\gamma}}{1-\gamma}, 
		& \text{if } \gamma>0 \text{ and } \gamma\neq 1, \\
		\log c, & \text{if } \gamma = 1.
	\end{cases}
\end{equation}

\begin{proposition}[Concavity of the consumption function]
	\label{pr:concavity}
	If the utility function is CRRA as defined in \eqref{eq:crra}, then the consumption function $c^*(w,z,\theta)$ is concave in $w$, and, for all $(z,\theta)\in \ZZ \times \Theta$, there exists $\alpha(z,\theta) \in [0,1]$  such that
	\begin{equation*}
		\lim_{w\to\infty} \frac{c^*(w,z,\theta)}{w} = \alpha(z,\theta).
	\end{equation*}
\end{proposition}

This result confirms that under CRRA preferences, households exhibit diminishing marginal propensity to consume as wealth increases, while consumption grows proportionally with wealth in the limit. The asymptotic linearity parameter $\alpha(z,\theta)$ captures the long-run consumption share of wealth and depends on the prevailing economic state and the household's beliefs.

\section{Quantitative Analysis}\label{s:quant}

This section applies the theoretical framework to design numerical algorithms, and then apply the theory and solution algorithm to a calibrated model of household consumption and savings under transition uncertainty. 

\subsection{Numerical Solution Method}\label{ss:alg}

We first outline the computational strategy. To that end, we explain how to efficiently construct grid points on the belief space, and then present the algorithms used to compute the optimal policy.

\subsubsection{Constructing Grid Points on the Belief Space}\label{sss:belief_grid}

To discretize the belief space efficiently, we construct a \emph{barycentric coordinate grid} over the $(N-1)$-dimensional probability simplex $\TT$. This approach generates evenly distributed grid points that satisfy the constraint $\sum_{i=1}^N \theta_i = 1$. 

For a given resolution parameter $H\in \NN$, we generate all nonnegative integer vectors $(h_1, \cdots, h_N)$ such that $\sum_{i=1}^{N} h_i = H$. Each such vector corresponds to a grid point
\begin{equation*}
	\theta = \left(
	    \frac{h_1}{H}, \dots, \frac{h_N}{H}
	\right) \in \TT.
\end{equation*}
The total number of grid points on $\TT$ is given by the binomial coefficient
\begin{equation*}
	L = \left(\begin{matrix}
		H+N-1 \\
		N-1
	\end{matrix}\right) 
	= \frac{\,(H+N-1)!\,}{(N-1)! H!},
\end{equation*}
which grows polinomially in $H$ for fixed $N$. In applications, we choose $H$ to balance approximation accuracy and computational cost. For example, when $N=3$ and $H=20$, we obtain $L=231$ grid points, providing a sufficiently fine discretization for smooth candidate functions.

\subsubsection{Computing the Optimal Policy}\label{ss:comp_opt_pol}

The preceding theory allows us to extend the endogenous grid point method of \citet{carroll2006method} to efficiently compute the optimal policy. In our framework, which features transition uncertainty and Bayesian learning, the method naturally generalizes. By fixing a grid for savings rather than wealth, it eliminates the need for root-finding during the time iteration process.

Let $\mathcal S_\theta = \{\theta_\ell\}_{\ell=1}^L$ denote the belief grid constructed above. We also specify a finite savings grid $\mathcal{S}_{s}:=\{s_{g}\}_{g=1}^{G}$, where $0=s_{1}<\cdots<s_{G}$. At each iteration, we update a candidate consumption function defined on $\mathcal{S}_{s} \times \ZZ\times \mathcal{S}_\theta$, and construct the corresponding wealth grid endogenously.

For each iteration $t \in \NN$, let $\mathcal{S}_w^{t}:=\left\{w^{t}_{g,\ell}(z)\right\}$ denote the endogenous wealth grid, where $w^{t}_{g,\ell}(z)$ is the wealth associated with $(s_{g},z,\theta_\ell)$. We initialize with $w^{0}_{g,\ell}(z)\equiv s_{g}$. Let $\cC(\mathcal{S}_w^t)$ denote the class of continuous piecewise linear functions $c:\SS \to \RR$ such that for each $z\in \ZZ$ and $\theta_\ell\in \mathcal S_\theta$: (1) $0\leq c(w,z,\theta_\ell)\leq w$ for all $w$; (2) $c(w,z,\theta_\ell)=w$ for $w \leq w^{t}_{1,\ell}(z)$; (3) $c(w,z,\theta_\ell)$ is linear in $w$ on each interval $[w^{t}_{g,\ell}(z),w^{t}_{g+1,\ell}(z)]$ for $g=1,...,G-1$; and (4) $c(w,z,\theta_\ell)$ is linearly extrapolated for $w>w^{t}_{G,\ell}(z)$. 

The endogenous grid algorithm proceeds as follows. 

\begin{algorithm}%\label{alg:egm}
	\caption{\, The endogenous grid algorithm}\label{alg:egm}
	\small 
	\begin{enumerate}[label=\textbf{Step \arabic*.}, leftmargin=*, ref=\arabic*, itemsep=1mm]
		\item(Initialization). Choose a convergence criterion $\varpi> 0$, a saving grid $\{s_{g}\}$, a belief grid $\{\theta_\ell\}$, and an initial policy  $c_{0}(w,z,\theta)$.
		
		\item\label{item:ega_update} (Policy Update). For each $t \in \NN$ and given $c_{t-1} \in \cC(\mathcal{S}_w^{t-1})$, update it as follows:
		\begin{enumerate}
			\item For each $(s_g,z,\theta_\ell) \in \mathcal S_s \times \ZZ\times \mathcal S_\theta$:
			\begin{enumerate}
				\item Compute future wealth $\hat w:=\hat R s_{g}+\hat Y$.
				
				\item Update beliefs using Bayes' rule and then project the posterior onto the nearest grid point $\hat \theta$ in $\mathcal S_\theta$.
				
				\item Compute the updated consumption
				\begin{equation*}
					\qquad \qquad \qquad \qquad
					\tilde c_{t}(s_g,z,\theta_\ell):=(u')^{-1}\left[ \min\left\{
					\EE_{z,\theta_\ell} \hat \beta \hat R (u'\circ c_{t-1}) 
					\big(\hat w, \hat Z,\hat \theta\big)
					\right\}, u'(0)
					\right].
				\end{equation*}
			\end{enumerate}
			
			\item Construct the endogenous wealth grid and associated consumption values:
			\begin{equation*}
				\qquad \qquad \qquad
				w^{t}_{g,\ell}(z):=s_{g}+\tilde c_{t}(s_{g},z,\theta_\ell) 
				\quad \text{and} \quad 
				c_{t}(w^{t}_{g,\ell}(z),z,\theta_\ell):=\tilde c_{t}(s_{g},z,\theta_\ell).
			\end{equation*}
			
			\item Extend $c_t$ to $c_t \in \cC(\mathcal{S}_w^t)$ by linear interpolation and extrapolation, and impose $c_{t}(w,z,\theta_\ell)=w$ for $w\leq w^{t}_{1,\ell}(z)$.
		\end{enumerate}
		
		\item (Convergence). Repeat Step~\ref{item:ega_update} until the sequence $\left\{c_{t}\big(w^{t}_{g,\ell}(z),z, \theta_\ell \big)\right\}$ converges within tolerance $\varpi$.
	\end{enumerate}
\end{algorithm}

\subsection{Quantitative Results and Mechanism Analysis}\label{ss:app}

We now present a quantitative application of the theory. The model features CRRA utility with coefficient of relative risk aversion $\gamma$ and a constant discount factor $\beta$. The exogenous state $Z_t \in \{0,1\}$ follows a two-state Markov chain, representing expansion ($Z_t=0$) and recession ($Z_t=1$). The labor income process is specified as
\begin{equation*}
	Y_t = Y(Z_t, \varepsilon_t^Y) = y(Z_t) \varepsilon_t^Y, 
	\quad \{\varepsilon_t^Y\} \iidsim LN(0, \sigma_Y^2).
\end{equation*}
Here, $y(Z_t)$ is a state-dependent persistent component, while $\epsilon_t^Y$ is a transitory multiplicative shock. The agent allocates a constant fraction $\alpha$ of savings to a risky asset and the remainder $1-\alpha$ to a risk-free asset. The risky log return depends on the state $Z_t$. The resulting gross return on the portfolio is
\begin{equation*}
	R_t = R (Z_t, \varepsilon_t^R) = \alpha R_f \text{e}^{\mu(Z_t) + \sigma(Z_t) \varepsilon_t^R} + (1-\alpha) R_f,
	\quad \{\varepsilon_t^R\} \iidsim N(0,1).
\end{equation*}
Here, $\mu(Z_t)$ and $\sigma(Z_t)$ are the state-contingent log risk premium and
volatility, and $R_f$ is the gross risk-free rate.

A central feature of the application is that the agent does not know the true transition probability for $\{Z_t\}$. The agent believes it is one of two candidates, $\pP = \{P_1,P_2\}$. The matrix $P_2$ is estimated using the monthly NBER recession indicator from January 1947 to January 2026, which we treat as the true data generating process in our simulations. To capture model uncertainty regarding regime persistence, we construct $P_1$ as a less persistent transition matrix relative to the highly persistent benchmark $P_2$. Under $P_1$, both expansions and recessions are shorter-lived, implying more frequent regime switches and greater macroeconomic instability. The resulting matrices are 
\begin{equation*}
	P_1 = \begin{bmatrix}
		0.8 & 0.2 \\
		0.3 & 0.7
	\end{bmatrix}
	\quad \text{and} \quad 
	P_2 = \begin{bmatrix}
		0.9855 & 0.0145 \\
		0.0968 & 0.9032
	\end{bmatrix}.
\end{equation*}
We set the other parameters as follows. We suppose that one period corresponds to one month and adopt standard values for the preference parameters: $\beta = \me^{-0.05/12}$ (5\% annual discounting) and $\gamma=2$. Regarding asset return, we set $\alpha=0.4$ for portfolio share, which is broadly in line with the moderate equity shares commonly used in life cycle portfolio choice models (see, e.g., \cite{cocco2005consumption}). Using the extended dataset of \citet{welch2008comprehensive} (monthly data from January 1947 to December 2024), we estimate the log risk-free rate as $\log R_f = 3.084 \times 10^{-4}$ (an annual rate 0.37\%), the state-dependent risk premia as $(\mu(0),\mu(1)) = (7.139,-1.735)\times 10^{-3}$, and the return volatility as $(\sigma(0), \sigma(1)) = (0.0391, 0.0577)$. 

The income process is calibrated to match the estimates of \citet{klein2013measuring}. Their study decomposes income into a persistent and a transitory component, and estimates a monthly frequency model using annual longitudinal household-level US income data from 1968--1997. In our implementation, we map their persistent AR(1) component into a two‑state Markov chain, preserving the essential dynamics while simplifying the computational structure. We set the variance of the transitory shock to their estimated value $\sigma_Y^2 = 0.5395$. We calibrate the persistent component so that the stationary mean and variance of our two-state chain match the corresponding moments of the persistent component reported in \citet{klein2013measuring}, yielding $(Y(0), Y(1)) = (1.8539, 0.0165)$.

All assumptions stated in Section~\ref{s:or} are satisfied in the present quantitative setting. In particular, the matrix
\begin{equation*}
	P^* = \begin{bmatrix}
		0.9855 & 0.0145 \\
		0.3 & 0.7
	\end{bmatrix}
\end{equation*}
is irreducible and monotone, and satisfies $P_1, P_2 \preceq P^*$ and $r(P^*D_\alpha)<1$ for $\alpha\in \{0,1\}$ under the calibrated parameters.\footnote{In 
	our calibration, state $0$ corresponds to expansion and state $1$ to recession in order to be consistent with the dataset, so $z_0$ is economically \emph{better} than $z_1$. Hence, the partial order $\preceq$ on $\ZZ$ is defined by $z_1 \preceq z_0$, and the monotonicity of $P_2$ is understood with respect to this ordering.} 
Hence, the uniform stability condition Assumption~\ref{a:unif_sr} holds. Theorem~\ref{t:opt} therefore guarantees the existence of a unique optimal consumption policy that can be computed via Algorithm~\ref{alg:egm}. Moreover, the monotonicity and concavity conditions underlying Propositions~\ref{pr:monotonea}--\ref{pr:concavity} are also satisfied, implying that the optimal policy is increasing and concave in wealth and continuous. These theoretical properties provide foundation for quantitative analysis.

The model is solved using Algorithm~\ref{alg:egm}. The endogenous grid for savings consists of $2000$ points spaced exponentially between $0$ and $1000$, with a median point at $150$ to ensure sufficient density in the empirically relevant wealth region. For the belief state $\theta$, which lies in a one‑dimensional probability simplex with $N=2$ candidate transition matrices, we employ a uniform grid of $100$ points. Moreover, the innovations $\{\epsilon_t^R\}$ and $\{\epsilon_t^Y\}$ are discretized using a 7-point Gauss–Hermite quadrature rule for numerical integration. The time iteration algorithm continues until the maximum absolute change in consumption across the grid falls below the tolerance $\varpi=10^{-4}$. 

\begin{figure}%[htb!]
	\includegraphics[width=\linewidth]{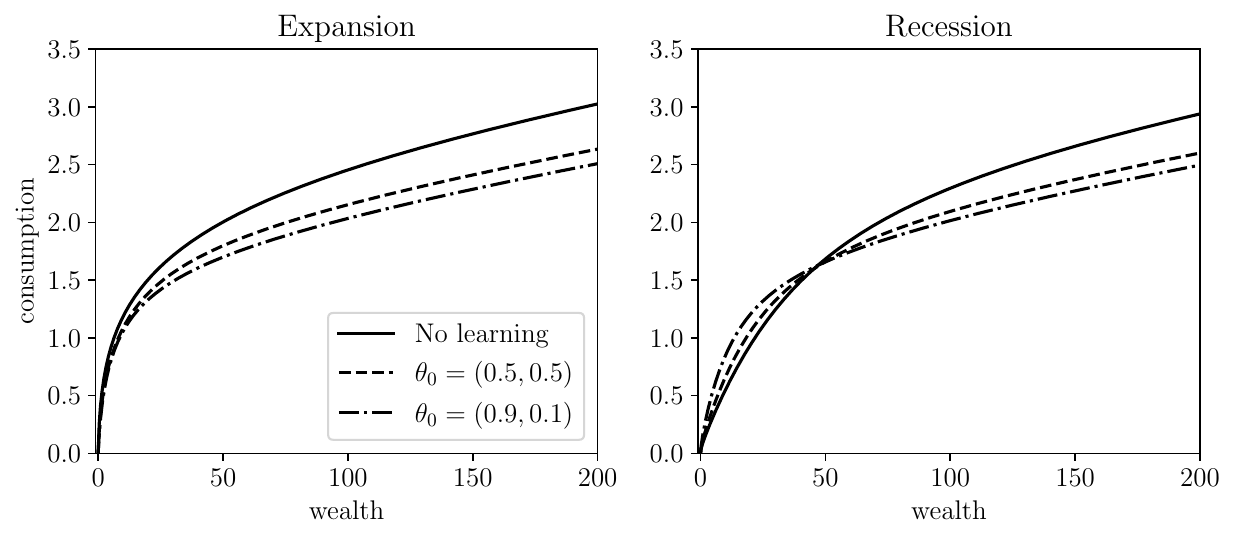}
	\caption{The optimal consumption policy}\label{fig:opt_pol}
\end{figure}

The quantitative results reveal a rich and fully dynamic interaction between transition uncertainty, learning, and household consumption and saving behavior. Figure~\ref{fig:opt_pol} shows how uncertainty about regime persistence reshapes the optimal consumption rule. In expansions, consumption under learning is uniformly lower across all wealth levels relative to the full-information benchmark. This outcome arises because agents assign positive probability to the less persistent transition matrix $P_1$, even though the true process is governed by $P_2$. As a result, households perceive a higher probability that the current expansion may end sooner than suggested by the true process. The elevated downside risk strengthens precautionary saving motives, leading households to reduce current consumption and accumulate larger buffer-stock savings.

In recessions, by contrast, the impact of learning becomes wealth-dependent. At low wealth levels, the possibility that the economy follows the less persistent matrix $P_1$ increases the perceived probability of a relatively rapid recovery. Because borrowing constraints bind most tightly for these households, the higher perceived likelihood of improvement partially relaxes precautionary pressure, resulting in slightly higher consumption relative to the full-information benchmark. At higher wealth levels, borrowing constraints are slack and intertemporal smoothing considerations dominate. In this case, the increased regime volatility implied by $P_1$ raises long-run income and return risk, strengthening precautionary motives and reducing consumption relative to the benchmark. The policy difference is therefore non-monotonic in wealth, reflecting the interaction between distorted beliefs and heterogeneous financial positions.

Figures~\ref{fig:cons_path}--\ref{fig:savings_path} translate these policy rules into a dynamic narrative, showing the simulated paths of aggregate consumption and savings.\footnote{The 
	consumption and savings paths are computed as $\EE_0 c^*(w_t,Z_t,\theta_t) $ and $\EE_0 [w_t - c^*(w_t,Z_t,\theta_t)]$, respectively, where $c^*$ is the optimal policy. Expectations are approximated via Monte Carlo simulation for $K=50,000$ sample paths:
	\begin{equation*}
		\EE_0 c^*(w_t,Z_t,\theta_t) 
		\approx \frac{1}{K} \sum_{i=1}^K 
		\EE_{w_{t-1}^i,Z_{t-1}^i,\theta_{t-1}^i} c^*(w_t,Z_t,\theta_t).
	\end{equation*}}
Initially, agents attach positive probability to the less persistent regime. This perceived instability amplifies precautionary behavior, producing a sizable initial consumption gap relative to the full-information economy. Quantitatively, consumption falls sharply in the first period, by 12.4\% and 17.2\% under two alternative prior specifications, while savings rise correspondingly. 

\begin{figure}%[htb!]
	\includegraphics[width=\linewidth]{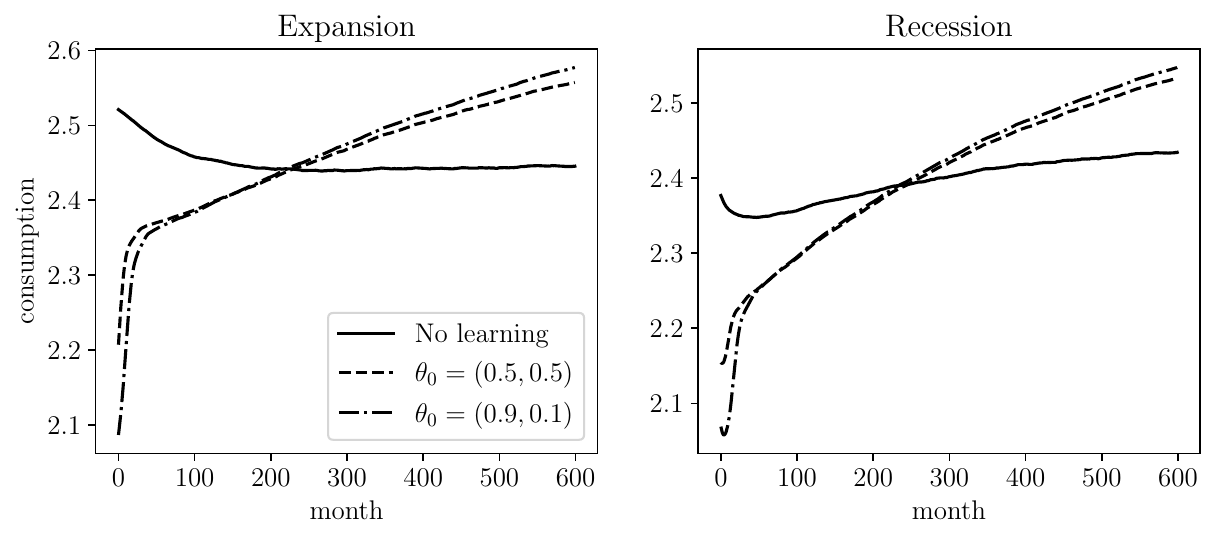}
	\caption{The expected consumption path}\label{fig:cons_path}
\end{figure}

\begin{figure}%[htb!]
	\includegraphics[width=\linewidth]{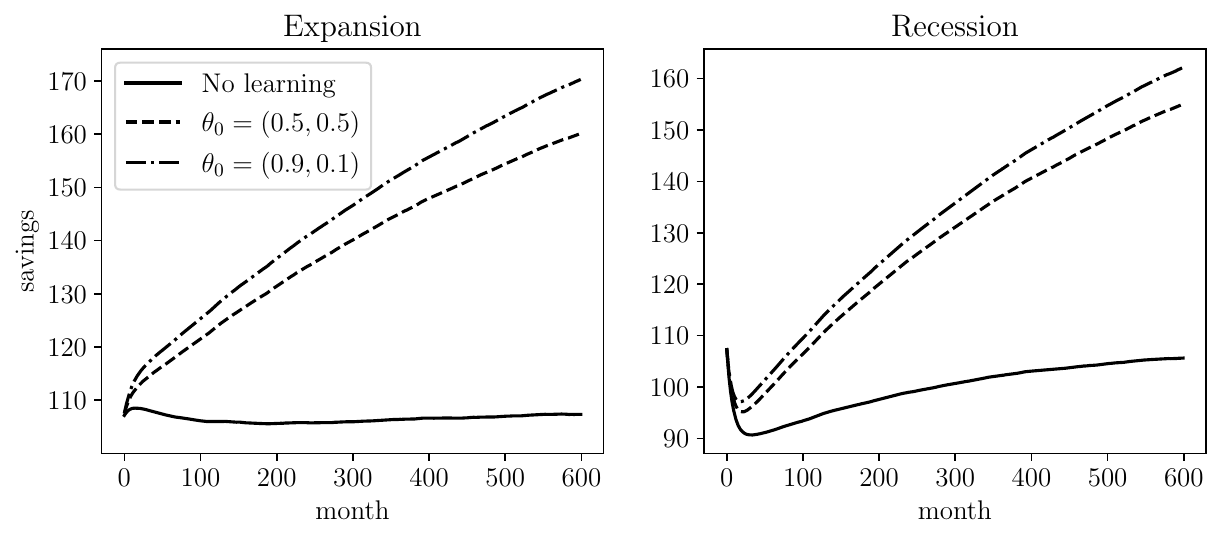}
	\caption{The expected savings path}\label{fig:savings_path}
\end{figure}

As the economy evolves under the true persistent transition probability $P_2$, posterior beliefs gradually concentrate on the high-persistence regime. The learning process therefore reduces perceived regime uncertainty and gradually attenuates precautionary motives. Consequently, the consumption gap narrows over time.

A key dynamic force, however, arises from the wealth accumulated during the early periods. The stronger precautionary saving motive initially generates a permanently higher wealth path. Once beliefs stabilize and precautionary pressure declines, this larger asset buffer becomes the dominant determinant of consumption dynamics. The learning economy eventually overtakes the full-information benchmark in terms of aggregate consumption despite its initially lower path. Savings remain persistently higher because the underlying wealth base is larger. This dynamic reversal illustrates the intertemporal reallocation induced by model uncertainty. Early consumption restraint leads to higher wealth accumulation, which subsequently supports higher long-run consumption.

The path of consumption volatility in Figure~\ref{fig:c_volatility_path} completes the mechanism by highlighting its welfare implications.\footnote{The 
	consumption volatility path is computed as 
	$\sqrt{\EE_0 [c^*(w_t,Z_t,\theta_t)^2] - [\EE_0 c^*(w_t,Z_t,\theta_t)]^2}$.} 
The early periods exhibit high and irregular volatility in the learning economy. This pattern reflects the sensitivity of posterior beliefs to limited initial data. When the observed history is short, small realizations of the aggregate state induce sizable revisions in posterior probabilities, producing abrupt changes in perceived transition dynamics. Consumption therefore responds not only to realized shocks but also to belief revisions, generating pronounced short-run volatility.

\begin{figure}%[htb!]
	\includegraphics[width=\linewidth]{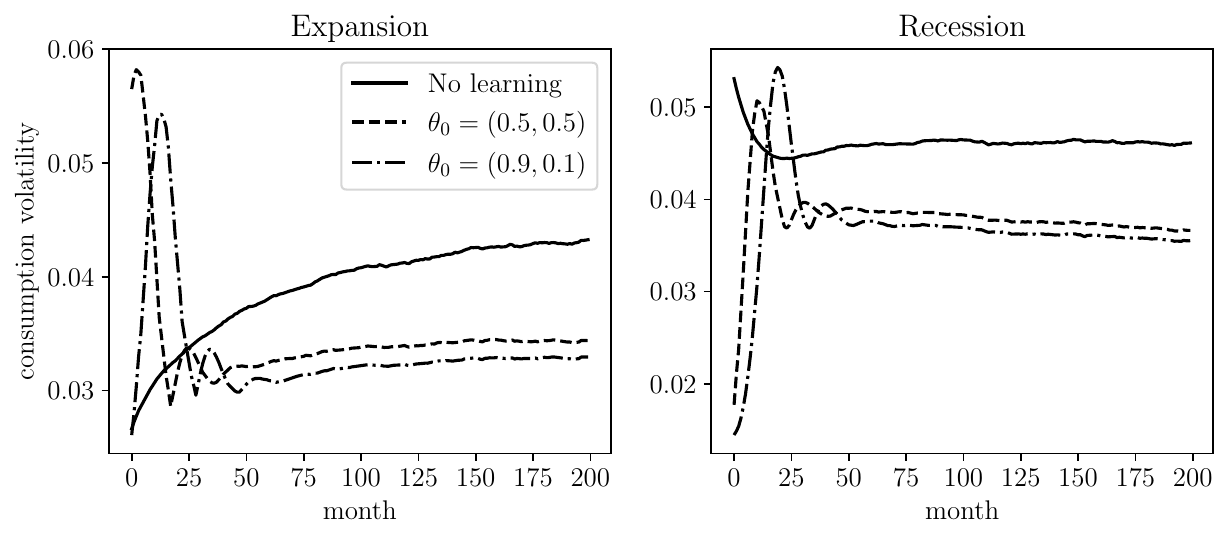}
	\caption{The expected consumption volatility path}\label{fig:c_volatility_path}
\end{figure}

As time progresses, two forces gradually compress volatility. First, posterior beliefs converge toward the true persistent regime, reducing fluctuations in perceived transition risk. Second, and more importantly, the learning economy enters this mature phase with a systematically higher wealth stock. The larger buffer of accumulated wealth dampens the sensitivity of consumption to income shocks, allowing households to absorb adverse realizations with smaller adjustments to current spending. In contrast, the full-information economy accumulates less wealth during the early phase and therefore remains more exposed to income fluctuations. As a result, long-run consumption volatility in the Bayesian learning economy falls below that of the benchmark.

Taken together, the numerical results reveal a coherent dynamic mechanism through which transition uncertainty affects household behavior. In the short run, belief dispersion and posterior updating amplify precautionary motives and generate additional volatility through belief revisions. Over time, however, the same precautionary behavior leads households to accumulate larger wealth buffers. As learning gradually resolves regime uncertainty, the economy transitions into a phase where the higher wealth stock dominates the dynamics. The result is stronger consumption smoothing and lower long-run consumption volatility. Transition uncertainty therefore has fundamentally dynamic consequences. It depresses consumption and increases volatility initially, but ultimately strengthens household balance sheets and improves long-run stability through endogenous wealth accumulation.

\appendix

\section{Proof of Section~\ref{s:or} Results}
\label{s:proof_or}

Let $(\Omega, \fF, \PP)$ be a fixed probability space on which all random 
variables are defined, and let $\EE$ denote expectations under $\PP$. The state 
process $\{Z_t\}$ and the innovation process $\{\epsilon_t\}$ introduced in 
\eqref{eq:law_exog} are defined on this space, with $\pi$ denoting the marginal 
distribution of $\{\epsilon_t\}$.

Expectations conditional on the initial state are denoted by $\EE_{w, z, \theta} [\,\cdot\,]$ when conditioning on $(w_0, Z_0, \theta_0) = (w, z,\theta)$, and by $\EE_{z, \theta}[\,\cdot\,]$ when conditioning on $(Z_0,\theta_0) = (z,\theta)$. Additionally, for any random variable $X$, its next period value is denoted by $\hat X$. By definition, for any random variable $X_t = X(w_t, Z_t, \theta_t, \epsilon_t)$,
\begin{equation*}
	\EE_{w,z,\theta} X = \int X(w,z,\theta,\epsilon) \pi (\diff \epsilon)
	\quad \text{and} \quad 
	\EE_{w,z,\theta} \hat X = 
	\sum_{\hat z\in \ZZ} P_{\theta}(z,\hat z) 
	\int X(\hat w, \hat z, \hat \theta, \hat \epsilon) \pi (\diff \hat \epsilon),
\end{equation*}
where $\hat \theta$ and $\hat w$ are defined according to the law of motion given by \eqref{eq:bayes} and \eqref{eq:w_hat}, respectively. Following this routine, for each $z\in \ZZ$ and $\alpha \in \{0,1\}$, we define 
\begin{equation}\label{eq:Dz}
	D_\alpha(z) := \EE_z \beta R^\alpha 
	= \EE_z \beta(z,\epsilon) R(z,\epsilon)^\alpha.
\end{equation}
Then $D_\alpha$ defined in \eqref{eq:D_mat} can be written equivalently as
\begin{equation*}
	D_\alpha = \diag\{D_\alpha(z_1), \cdots, D_\alpha(z_M)\}.
\end{equation*}
Moreover, for each $\theta \in \TT$, $i \in \{1, \dots, N\}$, and $\alpha\in\{0,1\}$, we define
\begin{equation*}
	K_i^\alpha = P_i D_\alpha, 
	\quad K_\theta^\alpha = P_\theta D_\alpha, 
	\quad \text{and} \quad 
	K_*^\alpha = P^* D_\alpha.
\end{equation*}
Throughout, we assume that Assumptions~\ref{a:utility}--\ref{a:unif_sr} hold. 

\begin{proposition}[Common upper eigenvector]
	\label{pr:unif_sr}
	For each $\alpha \in \{0,1\}$, there exists a strictly positive vector $x_\alpha \in \RR^M$ such that $K_i^\alpha x_\alpha < x_\alpha$ for all $i=1, \dots, N$.
\end{proposition}

\begin{proof}
	Fix $\alpha \in \{0,1\}$. Since $P^*$ is irreducible by Assumption~\ref{a:unif_sr} and $D_\alpha$ has strictly positive diagonal elements by Assumption~\ref{a:finite_exp}, the matrix $K_*^\alpha = P^* D_\alpha$ is irreducible and nonnegative. Let $\lambda_\alpha^* := r(K_*^\alpha)$. By the Perron–Frobenius theorem, there exists a unique (up to scalar multiples) strictly positive vector $x_\alpha\in \RR^M$ such that $K_*^\alpha x_\alpha = \lambda_\alpha^* x_\alpha$.
	
	Let $h:\ZZ \to \RR$ be any increasing function. Since $D_\alpha(z)$ defined in \eqref{eq:Dz} is increasing in $z$ and strictly positive by Assumption~\ref{a:finite_exp}, $D_\alpha(z) h(z)$ is also increasing in $z$. Since in addition $P^*$ is monotone by Assumption~\ref{a:unif_sr}, the function $K_*^\alpha h$ defined as follows is increasing on $\ZZ$: 
	\begin{equation*}
		(K^\alpha_* h)(z) := \sum_{\hat z\in \ZZ} 
		P^*(z,\hat z) D_\alpha(\hat z) h(\hat z).
	\end{equation*}
	Next, we initialize a sequence of functions $\{h_t\}_{t\geq 0}$ on $\ZZ$ by setting $h_0 \equiv 1$ (the constant function). We then define the sequence recursively by setting 
	\begin{equation*}
		h_{t+1} = \frac{1}{\lambda_\alpha^*} K^\alpha_* h_t.
	\end{equation*}
	Since $h_0$ is increasing, by induction, each $h_t$ is also increasing. Standard results from the theory of positive monotone operators guarantee that the sequence $\{h_t\}$ converges pointwise to a limiting function $h_\infty = \alpha x_\alpha$, where $\alpha >0$ is a constant and $x_\alpha$ is the Perron vector of $K_*^\alpha$. Furthermore, because each $h_t$ is increasing, the pointwise limit preserves monotonicity, so $x_\alpha$ must be increasing on $\ZZ$.
	
	Now consider the fact that $P_i \preceq P^*$ for all $i=1,\dots,N$. For any increasing function $h:\ZZ \to \RR$, it follows that $P_i h \leq P^* h$ pointwise for all $i=1,\dots,N$. Applying this inequality to the increasing vector $D_\alpha x_\alpha$, we obtain
	\begin{equation*}
		K_i^\alpha x_\alpha 
		= P_i D_\alpha x_\alpha 
		\leq P^* D_\alpha x_\alpha
		= K_*^\alpha x_\alpha 
		= \lambda_\alpha^* x_\alpha
		\quad \text{for all $i = 1, \dots, N$. }
	\end{equation*}
	Since $\lambda_\alpha^* <1$ by Assumption~\ref{a:unif_sr} and $x_\alpha \gg 0$, this implies that $K_i^\alpha x_\alpha < x_\alpha$ for all $i=1,\dots,N$, and the stated claim is verified.
\end{proof}

In what follows, for all $t\in \NN$ and $\alpha \in \{0,1\}$, we denote
\begin{equation}\label{eq:f_func}
	f_t^\alpha(z,\theta) 
	:= \EE_{z,\theta} \prod_{i=1}^{t} \beta_i R_i^\alpha 
	:= \EE_{z, \theta} \left[
	\beta_1R_1^\alpha \EE_{Z_1,\theta_1} \left[
	\beta_2 R_2^\alpha \cdots 
	\EE_{Z_{t-1},\theta_{t-1}} \beta_t R_t^\alpha
	\right]
	\right].
\end{equation}
In this function, $\theta_t$ is updated in each step following \eqref{eq:bayes}, the Bayes' rule. 

\begin{proposition}[Eventual contractivity]
	\label{pr:ev_contr}
	There exist an $n\in \NN$ and $\lambda\in(0,1)$ such that $f_t^\alpha(z,\theta) \leq \lambda^t$ for all $\alpha\in \{0,1\}$, $(z, \theta)\in \ZZ \times \TT$, and $t\geq n$. In particular, we have
	\begin{equation}\label{eq:finite_fsum}
		\sum_{t=0}^{\infty} f_t^\alpha (z,\theta) < \infty
		\quad \text{for all } \alpha\in\{0,1\}
		\text{ and } (z, \theta)\in \ZZ \times \TT.
	\end{equation}
\end{proposition}

\begin{proof}
	We first show that, for all $t\in \NN$,
	\begin{equation}\label{eq:exp_mat}
		f_t^\alpha(z,\theta)
		= \me_z^\top (K_\theta^\alpha 
		K_{\theta_1}^\alpha \cdots 
		K_{\theta_{t-1}}^\alpha) \1,
	\end{equation}
	where $\me_z^\top$ is the transpose of the $z$-th simple basis vector, i.e., a row vector with the $z$-th entry equal to $1$ and the other entries equal to $0$. Note that when $t=1$, 
	\begin{equation*}
		f_1^\alpha(z,\theta) 
		= \EE_{z, \theta} \beta_1 R_1^\alpha 
		= \sum_{\hat z\in \ZZ} P_\theta(z,\hat z) D_\alpha(\hat z) 
		= \me_z^\top P_\theta (D_\alpha \1) 
		= \me_z^\top K_\theta^\alpha \1.
	\end{equation*}
	Hence, \eqref{eq:exp_mat} holds for $t=1$. Suppose \eqref{eq:exp_mat} holds for arbitrary $t$. It remains to show that it holds for $t+1$. By the induction hypothesis and the Markov property,
	\begin{align*}
		f_{t+1}^\alpha (z,\theta) 
		&= \EE_{z,\theta} \beta_1 R_1^\alpha f_t^\alpha(Z_1, \theta_1) 
		= \sum_{\hat z\in \ZZ} P_\theta(z,\hat z) 
		D_\alpha(\hat z) f_t^\alpha(\hat z, \theta_1)  \\
		&= \sum_{\hat z\in \ZZ} P_\theta(z,\hat z) 
		D_\alpha(\hat z) \me_{\hat z}^\top 
		(K_{\theta_1}^\alpha \cdots K_{\theta_{t-1}}^\alpha) \1 
		= \me_z^\top K_\theta^\alpha 
		(K_{\theta_1}^\alpha \cdots K_{\theta_{t-1}}^\alpha) \1.
	\end{align*}
	Hence \eqref{eq:exp_mat} holds by induction.
	
	Fix $\alpha \in \{0,1\}$. By Proposition~\ref{pr:unif_sr}, there exists a strictly positive vector $x_\alpha \in \RR^M$ and $\eta \in (0,1)$ with $K_i^\alpha x_\alpha \leq \eta x_\alpha$ for all $i \in \{1, \dots, N\}$. Hence,
	\begin{equation*}
		\max_{i\in \{1,\dots,N\}} \max_{z\in \ZZ}
		\frac{(K_i^\alpha x_\alpha)(z)}{x_\alpha(z)} \leq \eta.
	\end{equation*}
	Then for all $\theta\in\TT$, we have
	\begin{equation}\label{eq:unif_ev}
		K_\theta^\alpha x_\alpha 
		= \sum_{i=1}^N \theta_i K_i^\alpha x_\alpha 
		\leq \sum_{i=1}^N \theta_i (\eta x_\alpha) = \eta x_\alpha.
	\end{equation}
	We define a vector norm on $\RR^M$ by
	\begin{equation*}
		\|v\| := \max_{z\in \ZZ} \frac{|v(z)|}{x_\alpha(z)},
		\qquad v \in \RR^M.
	\end{equation*}
	In particular, for all nonnegative vector $v$ in $\RR^M$, we have $v \leq \|v\| x_\alpha$. Multiplying both sides by $K_\theta^\alpha$ and applying \eqref{eq:unif_ev}, we obtain $K_\theta^\alpha v \leq \|v\| K_\theta^\alpha x_\alpha \leq \eta \|v\| x_\alpha$, hence $\|K_\theta^\alpha v\| \leq \eta \|v\|$. Next, we show that 
	\begin{equation}\label{eq:prodK_ub}
		\|K_\theta^\alpha K_{\theta_1}^\alpha 
		\cdots K_{\theta_{t-1}}^\alpha v \| 
		\leq \eta^t \|v\|
		\quad \text{for all } v \geq 0,\,  
		\theta \in \TT, 
		\text{ and } t \in \NN.
	\end{equation}
	We have shown that it holds for $t=1$. Now suppose \eqref{eq:prodK_ub} holds for arbitrary $t$. It remains to show that it holds for $t+1$. Let $\nu := K_{\theta_1}^\alpha \cdots K_{\theta_t}^\alpha v$. Then by the induction hypothesis, we have
	\begin{equation*}
		\|K_\theta^\alpha K_{\theta_1}^\alpha \cdots K_{\theta_t}^\alpha v \| 
		= \|K_\theta^\alpha \nu\| \leq \eta \|\nu\|
		\leq \eta^{t+1} \|v\|.
	\end{equation*}
	Hence, \eqref{eq:prodK_ub} holds by induction. Combining with \eqref{eq:exp_mat} and applying the dual norm inequality, we obtain
	\begin{equation*}
		f_t^\alpha(z,\theta) 
		\leq |\me_z^\top K_\theta^\alpha K_{\theta_1}^\alpha 
		\cdots K_{\theta_{t-1}}^\alpha \1|
		\leq \|\me_z\|_* 
		\|K_\theta^\alpha K_{\theta_1}^\alpha 
		\cdots K_{\theta_{t-1}}^\alpha \1 \| 
		\leq \eta^t \|e_z\|_* \|\1\|
	\end{equation*}
	for all $(z,\theta)\in \ZZ \times \TT$ and $t\in\NN$, where $\|\cdot\|_*$ is the dual norm of $\|\cdot\|$ satisfying
	\begin{equation*}
		\|\me_z\|_* := \sup_{\|v\|\leq 1} |\me_z^\top v|.
	\end{equation*}
	Since $\ZZ$ is finite-valued, we have 
	\begin{equation*}
		C := \max \left\{
		\max_{z\in \ZZ} \|e_z\|_* \|\1\|, 1
		\right\} < \infty.
	\end{equation*}
	Hence, there exists $n \in \NN$ such that $\lambda:= \eta C^{1/n} \in (0,1)$. Therefore, 
	\begin{equation*}
		f_t^\alpha(z,\theta) \leq (\eta C^{1/t})^t \leq \lambda^t
	\end{equation*} 
	for all $t\geq n$ and $(z,\theta)\in \ZZ \times \TT$. Since $\alpha$ is chosen arbitrarily, the first claim is verified. The second claim, namely, \eqref{eq:finite_fsum}, follows immediately from the first claim.
\end{proof}

\begin{lemma}[Finite expected discounted wealth]
	\label{lm:max_path}
	For the maximal asset path $\{ \tilde{w}_t \}$ defined by 
	\begin{equation}
		\label{eq:max_path}
		\tilde{w}_{t+1} = R_{t+1} \, \tilde{w}_t + Y_{t+1},
	\end{equation}
	where $ (\tilde{w}_0, \tilde{Z}_0,\tilde \theta_0) = (w,z,\theta) \in \SS$ is given, we have
	\begin{equation*}
		F(w, z,\theta) := \sum_{t = 0}^\infty \EE_{w,z,\theta} 
		\prod_{i=0}^t \beta_i \tilde{w}_t < \infty.
	\end{equation*}
\end{lemma}

\begin{proof}
	Iterating backward on \eqref{eq:max_path}, we can show that
	\begin{equation*}
		\tilde{w}_t = 
		\prod_{i=1}^t R_i w + \sum_{j=1}^t Y_j \prod_{i=j+1}^t R_i.
	\end{equation*}
	Taking expectation yields
	\begin{equation*}
		\EE_{w,z,\theta} \prod_{i=0}^t \beta_i \tilde{w}_t 
		= \EE_{z,\theta} \prod_{i=1}^t \beta_i R_i w + 
		\sum_{j=1}^t       
		\EE_{z,\theta}             
		\prod_{i=j+1}^t \beta_i R_i
		\prod_{k=0}^j \beta_k Y_j.   
	\end{equation*}
	Applying the monotone convergence theorem and the Markov property gives
	\begin{align*}
		F(w,z,\theta) &=
		\sum_{t=0}^{\infty} \EE_{z,\theta} \prod_{i=1}^t \beta_i R_i w +
		\sum_{t=0}^{\infty} \sum_{j=1}^t       
		\EE_{z,\theta} \prod_{i=j+1}^t \beta_i R_i 
		\prod_{k=0}^j \beta_k Y_j    \\
		&= \sum_{t=0}^{\infty} \EE_{z,\theta} \prod_{i=1}^t \beta_i R_i w +
		\sum_{j=1}^{\infty} \sum_{i=0}^\infty \EE_{z,\theta} 
		\prod_{k=0}^j \beta_k Y_j
		\prod_{\ell=1}^{i} \beta_{j+\ell} R_{j+\ell}    \\
		&= \sum_{t=0}^{\infty} \EE_{z,\theta} \prod_{i=1}^t \beta_i R_i w +
		\sum_{j=1}^{\infty} \EE_{z,\theta} \prod_{k=0}^j \beta_k \, Y_j \,
		\EE_{Z_j,\theta_j} \sum_{i=0}^\infty \prod_{\ell=1}^{i} \beta_{\ell} R_{\ell}.
	\end{align*}
	By Proposition~\ref{pr:ev_contr} and Assumption~\ref{a:finite_exp}, there exist constants $C_1, C_2, C_3 \in(0,\infty)$ such that
	\begin{equation*}
		\sum_{t=0}^{\infty} f_t^1(z,\theta) \leq C_1,
		\quad \sum_{t=0}^{\infty} f_t^0(z,\theta) \leq C_2,
		\quad \text{and} \quad 
		\max_{z\in \ZZ} \EE_z \beta Y \leq C_3. 
	\end{equation*}
	Hence, we obtain
	\begin{align*}
		F(w,z,\theta) &= \sum_{t=0}^{\infty} f_t^1(z,\theta) w + 
		\sum_{t=1}^{\infty} \EE_{z,\theta} \prod_{i=0}^{t} \beta_i Y_t 
		\sum_{k=0}^\infty f_k^1(Z_t, \theta_t) \\
		&\leq C_1 w + 
		    C_1 \sum_{t=1}^{\infty} \EE_{z,\theta} \prod_{i=0}^{t} \beta_i Y_t 
		= C_1 w + C_1 \sum_{t=1}^{\infty} \EE_{z,\theta} 
		\prod_{i=0}^{t-1} \beta_i \EE_{Z_t} \beta Y \\
		&\leq C_1 w + C_1 C_3 \sum_{t=0}^{\infty} f_t^0(z,\theta) 
		\leq C_1 w + C_1 C_2 C_3 < \infty.
	\end{align*}
	The stated claim is verified. 
\end{proof}

\begin{proof}[Proof of Proposition~\ref{pr:suff_foc}]
	Let $c$ be a policy in $\cC$ satisfying \eqref{eq:foc}. To see that any asset path generated by $c$ satisfies the transversality condition \eqref{eq:tvc}, note that $c\in \cC$ implies existence of a $B \in \RR$ such that 
	\begin{equation}\label{eq:bd_uprime}
		u'(w) \leq (u' \circ c)(w,z,\theta) \leq u'(w) + B
	\end{equation}
	for all $(w,z,\theta) \in \SS$. Therefore,
	\begin{equation}
			\label{eq:ineq_betu'ca}
			\EE_{w,z,\theta} \prod_{i=0}^t \beta_i \, (u' \circ c) (w_t, Z_t,\theta_t) w_t
			\leq \EE_{w, z,\theta} \prod_{i=0}^t \beta_i \, u'(w_t) w_t + 
			B \, \EE_{w, z,\theta} \prod_{i=0}^t \beta_i \, w_t .
		\end{equation}
	Regarding the first term on the right hand side of \eqref{eq:ineq_betu'ca}, 
	fix $A > 0$ and observe that
	\begin{align*}
			u'(w_t) w_t 
			&= u'(w_t) w_t \1 \{w_t \leq A\} + u'(w_t) w_t \1 \{w_t > A\}    \\
			&\leq A u'(w_t) + u'(A) w_t   
			\leq A u'(Y_t) + u'(A) \tilde{w}_t
		\end{align*}
	with probability one, where $\tilde{w}_t$  is the maximal path defined in \eqref{eq:max_path}. We 
	then have
	\begin{equation}
		\label{eq:ineq_betEu'a}
		\EE_{w,z,\theta} \prod_{i=0}^t \beta_i \, u'(w_t) w_t 
		\leq A \EE_{z,\theta} \prod_{i=0}^t \beta_i \, u'(Y_t) 
		+ u'(A) \EE_{w,z,\theta} \prod_{i=0}^t \beta_i \, \tilde{w}_t.
	\end{equation} 
	By Assumption~\ref{a:finite_exp} and the Markov property, there exists $C_0\in(0,\infty)$ such that 
	\begin{equation*}
		A \EE_{z,\theta} \prod_{i=0}^t \beta_i u'(Y_t) 
		= A \EE_{z,\theta} \prod_{i=0}^{t-1} \beta_i \EE_{Z_t} \beta u'(Y) 
		\leq C_0 A f_{t-1}^0(z,\theta),
	\end{equation*}
	and the last expression converges to zero as $t \to \infty$ by Proposition~\ref{pr:ev_contr}. The second term in \eqref{eq:ineq_betEu'a} also converges to zero by Lemma~\ref{lm:max_path}. Hence $\EE_{w,z,\theta} \prod_{i=0}^t \beta_i \, u'(w_t) w_t\to 0$ as
	$t \to \infty$, which, combined with \eqref{eq:ineq_betu'ca} and another
	application of Lemma~\ref{lm:max_path}, gives our desired result.
\end{proof}

\begin{lemma}
	\label{lm:complete}
	$(\cC, \rho)$ is a complete metric space.
\end{lemma}

\begin{proof}
	The proof is a straightforward extension of Proposition~4.1 of \cite{li2014solving} and thus omitted. 
	A full proof is available from the authors upon request.
\end{proof}

\begin{lemma}[Well-defined operator]
	\label{lm:welldef_T}
	For all $c \in \cC$ and $(w,z,\theta) \in \SS$, there exists a unique $\xi \in [0,w]$ that solves \eqref{eq:tio}.
\end{lemma}

\begin{proof}
	For given $c \in \cC$, we can rewrite \eqref{eq:tio} as 
	\begin{equation}\label{eq:psi_c}
		u'(\xi) 
		= \psi_c(\xi, w,z,\theta) 
		:= \min\{\max \{g_c(\xi, w,z,\theta), u'(w)\}, u'(0)\},
	\end{equation}
	where $g_c$ is a function on 
	\begin{equation}\label{eq:G_set}
		G:= \{(\xi,w,z,\theta)\in [0,\infty) \times \SS: 0 \leq \xi \leq w \}
	\end{equation}
	defined by
	\begin{equation}\label{eq:g_c}
		g_c(\xi,w,z,\theta) 
		:= \EE_{z,\theta} \hat \beta \hat R 
		(u'\circ c)(\hat R(w-\xi) + \hat Y, \hat Z, \hat \theta).
	\end{equation}
	Fix $c \in \cC$ and $(w,z,\theta) \in \SS$. Since $c \in \cC$, the map $\xi \mapsto \psi_c(\xi,w,z,\theta)$ is increasing. Additionally, because $u'$ is strictly decreasing, \eqref{eq:tio} can have at most one solution, and uniqueness is verified.
	
	To establish existence, we apply the intermediate value theorem. It suffices to verify the following three conditions:
	\begin{enumerate}
		\item[(a)] the map $\xi \mapsto \psi_c(\xi, w, z,\theta)$ is continuous on $[0, w]$,
		\item[(b)] there exists a $\xi \in [0,w]$ such that 
		$u'(\xi) \leq \psi_c(\xi, w, z,\theta)$, and
		\item[(c)] there exists a $\xi \in [0,w]$ such that 
		$u'(\xi) \geq \psi_c(\xi, w, z,\theta)$.
	\end{enumerate}
	For part (a), it is sufficient to show that $\xi \mapsto g_c(\xi,w,z,\theta)$ is continuous on $[0,w]$. To this end, fix $\xi \in [0,w]$ and let $\xi_n \to \xi$. Since $c\in \cC$, there exists a constant $B \in \RR_+$ such that 
	\begin{equation*}
		u'(w) \leq (u'\circ c)(w,z,\theta) \leq u'(w) + B
		\quad \text{for all } (w,z,\theta) \in \SS.
	\end{equation*}
	By the monotonicity of $u'$, we have
	\begin{equation}\label{eq:uppbd_ruprmc}
		\hat{\beta} \hat{R} 
		    (u'\circ c) (\hat{R} \left(w - \xi \right) + \hat{Y}, 
		    \hat{Z}, \hat \theta)
		\leq \hat{\beta} \hat{R} 
		    (u'\circ c)(\hat{Y}, \hat{Z}, \hat\theta)
		\leq \hat{\beta} \hat{R} u'(\hat{Y}) 
		    + \hat{\beta} \hat{R} B.
	\end{equation}
	Furthermore, by Assumptions~\ref{a:finite_exp}--\ref{a:unif_sr}, we know that, for all $(z,\theta)\in \ZZ\times\TT$,
	\begin{equation*}
		\EE_{z,\theta} \hat \beta \hat R < \infty
		\quad \text{and} \quad
		\EE_{z,\theta} \hat \beta \hat R u'(\hat Y) < \infty.
	\end{equation*}
	Hence, by the dominated convergence theorem and the continuity of $c$, we 
	can conclude that $g(\xi_n,w,z,\theta) \to g(\xi,w,z,\theta)$. This proves that $\xi \mapsto \psi_c(\xi, w, z,\theta)$ is continuous.
	
	Part~(b) holds trivially by setting $\xi=w$. When $u'(0) < \infty$, part~(c) holds trivially by setting $\xi=0$. When $u'(0) = \infty$, we have $u'(\xi) \to \infty$ as $\xi \to 0$. Since in addition $\xi \mapsto \psi_c(\xi, w, z,\theta)$ is increasing and always finite (since it is continuous as shown in the previous paragraph), part~(c) holds when $\xi$ gets sufficiently close to zero. The proof is now complete.
\end{proof}

\begin{lemma}[Self-mapping]
	\label{lm:self_map}
	We have $Tc \in \cC$ for all $c \in \cC$.
\end{lemma}

\begin{proof}
	Fix $c \in \cC$ and let $\psi_c$ and $g_c$ be defined as in \eqref{eq:psi_c} and \eqref{eq:g_c}, respectively.
	
	\textbf{Step~1.} We show that $Tc$ is continuous. To apply a standard 
	fixed point parametric continuity result, such as Theorem~B.1.4 of 
	\cite{stachurski2009economic}, we first demonstrate that $\psi_c$ is jointly continuous on the set $G$ defined in \eqref{eq:G_set}. This will hold if $g_c$ is jointly continuous on $G$. For any $\{ (\xi_n, w_n, z_n,\theta_n) \}$ and $(\xi, w, z,\theta)$ in $G$ with $(\xi_n, w_n, z_n,\theta_n) \to (\xi, w, z,\theta)$, we need to show that $g_c(\xi_n, w_n, z_n,\theta_n) \to g_c(\xi, w, z,\theta)$. To this end, let
	\begin{align*}
		&h_1 (\xi, w, \hat{Z}, \hat\theta, \hat{\epsilon}), \,
		h_2 (\xi, w, \hat{Z}, \hat\theta, \hat{\epsilon})  \\
		&:= \hat{\beta} \hat{R} 
		(u'(\hat Y) + B) \pm 
		\hat{\beta} \hat{R} 
		(u'\circ c)(\hat{R} \left( w - \xi \right) + \hat{Y}, \hat{Z},\hat \theta),
	\end{align*}
	where $\hat{\beta} := \beta (\hat{Z}, \hat{\epsilon})$, 
	$\hat{R} := R (\hat{Z}, \hat{\epsilon})$, and 
	$\hat{Y} := Y (\hat{Z}, \hat{\epsilon})$ as defined in 
	\eqref{eq:law_exog}. Then $h_1$ and $h_2$ are continuous in 
	$(\xi, w, \hat{Z})$ by the continuity of $c$  and nonnegative by 
	\eqref{eq:uppbd_ruprmc}. By the Fatou's lemma and Theorem~1.1 of \cite{feinberg2014fatou}, we have
	\begin{align*}
		\int \sum_{\hat z \in \ZZ}
		h_i ( \xi, w, \hat{z}, \hat\theta, \hat{\epsilon})
		P_\theta(z,\hat{z}) \pi(\diff \hat{\epsilon}) 
		&\leq \int \liminf_{n \to \infty}
		\sum_{\hat z \in \ZZ}
		h_i ( \xi_n, w_n, \hat{z}, \hat\theta, \hat{\epsilon})
		P_{\theta_n}(z_n, \hat{z})
		\pi(\diff \hat{\epsilon})  \\
		& \leq \liminf_{n \to \infty}
		\int 
		\sum_{\hat z \in \ZZ}
		h_i ( \xi_n, w_n, \hat{z}, \hat\theta, \hat{\epsilon})
		P_{\theta_n}( z_n, \hat z)
		\pi (\diff \hat{\epsilon}),
	\end{align*}
	which implies 
	\begin{align*}
		&\liminf_{n \to \infty} \left[
		\pm \EE_{z_n,\theta_n} \hat{\beta} \hat{R} 
		    (u'\circ c)(\hat{R} (w_n - \xi_n) + \hat{Y}, \hat{Z},\hat \theta)
		\right]  \\
		&\geq \left[ 
		\pm \EE_{z,\theta} \hat{\beta} \hat{R}
		    (u'\circ c)(\hat{R} (w - \xi) + \hat{Y}, \hat{Z},\hat \theta)
		\right].
	\end{align*}
	This shows that $g_c$ is continuous, since the inequality above is equivalent to 
	\begin{equation*}
		\liminf_{n \to \infty} g_c(\xi_n, w_n, z_n,\theta_n)   
		\geq g_c(\xi, w, z,\theta) \geq   
		\limsup_{n \to \infty} g_c(\xi_n, w_n, z_n,\theta_n).
	\end{equation*}
	Hence, $\psi_c$ is continuous on $G$, as required. Moreover, since $\xi 
	\mapsto \psi_c(\xi, w, z,\theta)$ takes values in the closed interval 
	\begin{equation*}
		I(w,z,\theta) := \begin{cases}
			\left[
				u'(w), 
				u'(w) + \EE_{z,\theta} \hat{\beta} \hat{R} (u'(\hat{Y}) + B)
			\right], & \text{if } u'(0) = \infty,  \\
			\left[
			    u'(w), u'(0)
			\right], & \text{otherwise.}
		\end{cases}
	\end{equation*}
	In both cases, the correspondence $(w, z) \mapsto I(w,z)$ is nonempty, compact-valued, and continuous, Theorem~B.1.4 of \cite{stachurski2009economic} implies that $Tc$ is continuous on $\SS$.
	
	\textbf{Step 2.} We show that $Tc$ is increasing in $w$. Suppose that for 
	some $(z,\theta) \in \ZZ\times \TT$ and $w_1, w_2 \in (0, \infty)$ with $w_1 < w_2$, we have 
	$\xi_1 := Tc (w_1,z,\theta) > Tc (w_2,z,\theta) =: \xi_2$. Since $c$ is increasing in $w$ 
	by assumption, $\psi_c$ is increasing in $\xi$ and decreasing in $w$. Thus, 
	we have $u'(\xi_1) < u'(\xi_2) = \psi_c(\xi_2, w_2, z,\theta) \leq \psi_c(\xi_1, w_1, z,\theta) = u'(\xi_1)$, which is a contradiction.
	
	\textbf{Step 3.} By Lemma~\ref{lm:welldef_T}, we have $Tc(w,z,\theta) \in [0,w]$ for all $(w,z,\theta) \in \SS$.
	
	\textbf{Step 4.} Through some simple algebra, we can show that
	\begin{align*}
		\left| u'[Tc(w,z,\theta)] - u'(w) \right| 
		&\leq \min \left\{ 
		    \EE_{z,\theta} \hat{\beta} \hat{R} 
		    (u'\circ c) (\hat{R}(w - Tc(w,z,\theta)) + \hat{Y}, 
		    \hat{Z}, \hat \theta), u'(0)
		\right\} \\
		&\leq \min\left\{
		    \EE_{z,\theta} \hat{\beta} \hat{R} (u'(\hat{Y}) + B), u'(0)
		\right\}
	\end{align*}
	for all $(w,z,\theta) \in \SS$. The last term is finite by 
	Assumptions~\ref{a:finite_exp}--\ref{a:unif_sr}.
\end{proof}

To prove Theorem~\ref{t:opt}, let $\hH$ denote the set of all continuous functions $h:\SS \to \RR_+$ such that each $h$ is decreasing in its first argument and $w \mapsto h(w,z,\theta)-u'(w)$ is bounded and nonnegative. For any $h \in \hH$, we define $\tilde{T} h(w,z,\theta)$ as the value $\kappa$ that solves
\begin{equation}\label{eq:kappa}
	\kappa = \min \left\{
		\max \left\{ 
		    \EE_{z,\theta} \hat{\beta} \hat{R} \,
		    h (\hat{R} \,[w - (u')^{-1}(\kappa)] + \hat{Y}, 
		    \hat{Z}, \hat\theta), u'(w) 
		\right\}, u'(0)
	\right\},
\end{equation}
Moreover, we consider $H: \cC \to \hH$ defined by 
$Hc(w,z,\theta) := (u'\circ c)(w,z,\theta)$. The next lemma implies that $\tilde{T}$ is a well-defined self-map on $\hH$, as well as topologically conjugate to $T$.

\begin{lemma}[Topological conjugacy]
	\label{lm:conjug}
	The operator $\tilde{T} \colon \hH \to \hH$ and satisfies 
	$\tilde{T} H = H T$ on $\cC$.
\end{lemma}

\begin{proof}
	Pick any $c \in \cC$ and $(w,z,\theta) \in \SS$. Let $\xi := Tc(w,z,\theta)$. By definition, $\xi$ solves
	\begin{equation}
		\label{eq:Tc_eq}
		u'(\xi) = \min\left\{
			\max \left\{ 
				\EE_{z,\theta} \hat{\beta} \hat{R}
				(u'\circ c)(\hat{R} \left(w - \xi \right) + \hat{Y}, 
				\hat{Z}, \hat\theta), u'(w) 
			\right\}, u'(0)
		\right\}.
	\end{equation}
	We need to show that $HTc$ and $\tilde{T} Hc$ evaluate to the same number 
	at $(w,z,\theta)$. In other words, we need to verify that $u'(\xi)$ is the 
	solution to
	\begin{equation*}
		\kappa = \min \left\{
			\max \left\{ 
				\EE_{z,\theta} \hat \beta \hat R 
				(u'\circ c)(\hat{R} \, [w - (u')^{-1} (\kappa)] + \hat Y, 
				\hat{Z}, \hat\theta), u'(w) 
			\right\}, u'(0)
		\right\}.
	\end{equation*}
	But this follows immediately from \eqref{eq:Tc_eq}. Hence, we have shown that
	$\tilde{T} H = H T$ on $\cC$. Since $H \colon \cC \to \hH$ is a bijection,
	we have $\tilde{T} = HT H^{-1}$. Moreover, Proposition~\ref{lm:self_map} 
	ensures that $T \colon \cC \to \cC$, and hence $\tilde{T} \colon \hH \to \hH$. 
	This completes the proof.
\end{proof}

\begin{lemma}[Order preserving]
	\label{lm:order_preserving}
	For all $h_1, h_2\in \hH$ with $h_1\leq h_2$, we have
	$\tilde T h_1 \leq \tilde T h_2$.
\end{lemma}

\begin{proof}
	Let $h_1,h_2 \in \hH$ with $h_1\leq h_2$. Suppose to the contrary that there exists $(w,z,\theta) \in \SS$ such that $\kappa_1 := \tilde Th_1(w,z,\theta) > Th_2(w,z,\theta) =:\kappa_2$. Then by the monotonicity of elements in $\hH$, we have
	\begin{align*}
		\kappa_1 &= \min\left\{
			\max \left\{
			    \EE_{z,\theta} \hat \beta \hat R h_1\left( 
			        \hat R [w - (u')^{-1}(\kappa_1)] + \hat Y, 
			        \hat Z, \hat \theta 
			    \right), u'(w)
			\right\}, u'(0)
		\right\}  \\
		&\leq \min\left\{
			\max \left\{
			    \EE_{z,\theta} \hat \beta \hat R h_1\left( 
			        \hat R [w - (u')^{-1}(\kappa_2)] + \hat Y, 
			        \hat Z, \hat \theta 
			    \right), u'(w)
			\right\}, u'(0)
		\right\}  \\
		&\leq \min\left\{
			\max \left\{
			    \EE_{z,\theta} \hat \beta \hat R h_2\left( 
			        \hat R [w - (u')^{-1}(\kappa_2)] + \hat Y, 
			        \hat Z, \hat \theta 
			    \right), u'(w)
			\right\}, u'(0)
		\right\} = \kappa_2.
	\end{align*}
	This is a contradiction. Hence the stated claim holds.
\end{proof}

\begin{lemma}[Discounting]\label{lm:discounting}
	Recall $f^\alpha_t (z,\theta)$ defined in \eqref{eq:f_func}.
	For all $h\in \hH$, $\gamma \in \RR_+$, and $k\in \NN$, we have
	\begin{equation}\label{eq:discounting}
		\tilde T^k (h+\gamma)(w,z,\theta) \leq \tilde T^k h(w,z,\theta) + \gamma f_k^1(z,\theta).
	\end{equation}
\end{lemma}

\begin{proof}
	For given functions $g,h:\SS \to \RR$, we define
	\begin{align*}
		&\psi(w,z,\theta; h, g) \\
		&:= \min \left\{
		\max\left\{
		\EE_{z,\theta} \hat \beta \hat R
		h(\hat R[w-(u')^{-1}(g(w,z,\theta))] + \hat Y, \hat Z,\hat \theta),
		u'(w)
		\right\}, u'(0)
		\right\}.
	\end{align*}
	Then $\psi$ is increasing in $h$ and decreasing in $g$. Fix $h\in \hH$ and $\gamma\in \RR_+$ and let $h_\gamma(w,z,\theta) := h(w,z,\theta) + \gamma$. By the definition of $\tilde T$, we have 
	\begin{align*}
		\tilde{T}h_\gamma (w,z,\theta) 
		&= \psi(w,z,\theta; h_\gamma, \tilde T h_\gamma) 
		\leq \psi(w,z,\theta; h, \tilde T h_\gamma) + \gamma \EE_{z,\theta}\beta_1 R_1  \\
		&\leq \psi(w,z,\theta; h, \tilde T h) + \gamma \EE_{z,\theta}\beta_1 R_1.
	\end{align*}
	Here, the first inequality is elementary, while the second inequality follows from Lemma~\ref{lm:order_preserving} and the fact that $h \leq h_\gamma$. Hence \eqref{eq:discounting} holds for $k=1$. Suppose it holds for arbitrary $k$. It remains to check that it holds for $k+1$. By the induction hypothesis, the monotonicity of $\tilde T$, and the Markov property, 
	\begin{align*}
		\tilde T^{k+1} h_\gamma(w,z,\theta) 
		&= \psi(w,z,\theta; \tilde T^k h_\gamma, \tilde T^{k+1} h_\gamma) 
		\leq \psi(w,z,\theta; \tilde T^k h + \gamma f_k^1, \tilde T^{k+1} h_\gamma) \\
		&\leq \psi(w,z,\theta; \tilde T^k h, \tilde T^{k+1} h_\gamma) 
		+ \gamma \EE_{z,\theta} \beta_1 R_1 
		\EE_{Z_1,\theta_1} f_k^1(Z_1, \theta_1)\\
		&\leq \psi(w,z,\theta; \tilde T^k h, \tilde T^{k+1} h) 
		+ \gamma \EE_{z,\theta} \beta_1 R_1 \cdots \beta_{k+1} R_{k+1} \\
		&= \tilde T^{k+1} h(w,z,\theta) + \gamma f_{k+1}^1(z,\theta).
	\end{align*}
	Hence \eqref{eq:discounting} is verified by induction.
\end{proof}

\begin{proof}[Proof of Theorem \ref{t:opt}]
	For $h_1,h_2\in \hH$, we define 
	\begin{equation*}
		d(h_1,h_2) 
		:= \sup_{(w,z,\theta) \in \SS}
		\left|h_1(w,z,\theta) - h_2(w,z,\theta)\right|.
	\end{equation*}
	Based on Proposition~\ref{pr:ev_contr}, Lemmas~\ref{lm:order_preserving}--\ref{lm:discounting}, and the sufficient conditions for contraction mapping proposed by \cite{blackwell1965discounted}, there exist an $n\in \NN$ and $\lambda \in (0,1)$ such that $\tilde T^n$ is a contraction mapping on $(\hH, d)$. In particular, $d (\tilde T^n h_1, \tilde T^n h_2) \leq \lambda d(h_1, h_2)$ for all $h_1,h_2\in \hH$, and there exists a 
	unique fixed point $h^*\in \hH$ of $\tilde{T}$. From the topological 
	conjugacy result in Lemma~\ref{lm:conjug}, we have $\tilde{T} = H T 
	H^{-1}$, so there exists a unique fixed point $c^*\in \cC$ of $T$, and 
	claim~(1) is verified. Claim~(2) follows immediately from claim~(1) and Proposition~\ref{pr:suff_foc}.
	
	To see that claim~(3) holds, based on the Banach contraction mapping theorem, it suffices to show that $\rho(T^n c, T^n d) \leq \lambda \rho(c,d)$ for all $c,d\in \cC$ for some $\lambda\in(0,1)$. To this end, pick any $c, d\in \cC$. Based on the definition of $\rho$, the topological conjugacy result in Lemma~\ref{lm:conjug}, and the contraction mapping property of $\tilde T^n$, we have
	\begin{equation*}
		\rho(T^nc, T^n d) 
		= d (HT^n c, HT^n d)
		= d (\tilde T^n Hc, \tilde T^n Hd)
		\leq \lambda d(Hc, Hd)
		= \lambda \rho(c,d).
	\end{equation*}
	Hence, claim~(3) is verified. The proof is now complete.
\end{proof}

\section{Proof of Section~\ref{s:ifp_properties} Results}

Through some simple algebra, we can verify the following fundamental result, which is frequently used in subsequent proofs.

\begin{lemma}\label{lm:basic}
	Let $a, b_1, b_2$, and $b_3$ be real numbers satisfying $a = \min\{\max\{b_1,b_2\}, b_3\}$. The following statements are true:
	\begin{enumerate}
		\item If $a > b_2$, then $a = \min\{b_1, b_3\}$.
		\item If $a < b_3$, then $a = \max\{b_1, b_2\}$.
		\item If $b_2 < a < b_3$, then $a = b_1$.
	\end{enumerate}
\end{lemma}

\begin{proof}[Proof of Proposition~\ref{pr:monotonea}]
	We define
	\begin{equation}\label{eq:cC_0}
		\cC_0 = \left\{
		c \in \cC \colon 
		w \mapsto w - c(w,z,\theta) 
		\text{ is increasing for all } 
		(z,\theta) \in \ZZ \times \Theta
		\right\}.
	\end{equation}
	Since $T$ is an eventual contraction on the complete metric space $(\cC,\rho)$ by Theorem~\ref{t:opt}, to prove the stated claim, it suffices to show that $\cC_0$ is a closed subset of $\cC$ and that $Tc \in \cC_0$ for all $c \in \cC_0$.
	
	To see that $\cC_0$ is closed, let $\{c_n\}$ be a sequence in $\cC_0$ and suppose $c \in \cC$ satisfies $\rho(c_n, c) \to 0$. We claim that $c \in 
	\cC_0$. To see this, for each $n$, the map $w \mapsto w - c_n(w,z,\theta)$ is increasing for all $z$, and the convergence $\rho(c_n,c) \to 0$ implies pointwise convergence $c_n (w,z,\theta) \to c(w,z,\theta)$ for all $(w,z,\theta) \in \SS$. Thus the monotonicity property is preserved in the limit, so $c \in \cC_0$.
	
	Fix $c \in \cC_0$. We now show that $\xi := Tc \in \cC_0$. By 
	Lemma~\ref{lm:self_map}, $\xi \in \cC$, hence it suffices to show 
	that $w \mapsto w - \xi(w,z,\theta)$ is increasing. Suppose not. Then there exist $(z,\theta) \in \ZZ\times \Theta$ and $w_1, w_2 \in [0, \infty)$ such that $w_1 < w_2$ and $w_1 - \xi(w_1, z,\theta) > w_2 - \xi (w_2, z,\theta)$. Since $w_1 - \xi (w_1, z,\theta) \geq 0$, $w_2 - \xi (w_2, z,\theta) \geq 0$ and $\xi(w_1,z,\theta) \leq \xi (w_2, z,\theta)$ by Lemma~\ref{lm:self_map}, we must have $\xi(w_1, z,\theta) < w_1$ and  $\xi(w_1, z,\theta) < \xi(w_2, z,\theta)$. However, by the definition of $T$, the concavity of $u$, and Lemma~\ref{lm:basic}, this implies
	\begin{align*}
		(u'\circ \xi) (w_1, z,\theta) 
		&= \min \left\{
		    \EE_{z,\theta} \hat{\beta} \hat{R} 
		    (u'\circ c)(\hat{R} [w_1 - \xi(w_1,z,\theta)] + \hat{Y}, 
		    \hat{Z},\hat{\theta}), u'(0)
		\right\} \\
		&\leq \min\left\{
		    \EE_{z,\theta} \hat{\beta} \hat{R} 
		    (u'\circ c)(\hat{R} [w_2 - \xi(w_2,z,\theta)] + \hat{Y}, 
		    \hat{Z},\hat \theta), u'(0)
		\right\} \\
		&\leq (u'\circ \xi) (w_2, z,\theta),
	\end{align*}
	which gives $\xi(w_1, z,\theta) \geq \xi(w_2, z,\theta)$, yielding a contradiction, hence $w \mapsto w - \xi (w,z,\theta)$ is increasing and $T$ is a self-map on $\cC_0$. The proof is now complete.
\end{proof}

\begin{proof}[Proof of Proposition~\ref{pr:monotoneY}]
	Let $T_j$ be the time iteration operator for the income process $j$
	as defined in Lemma~\ref{lm:self_map}. It suffices to show $T_1c
	\leq T_2c$ for all $c \in \cC$. To see this, by the monotonicity of $T_j$, 
	we have $T_jc_1 \leq T_jc_2$ whenever $c_1 \leq c_2$. Thus, if $T_1c \leq 
	T_2c$ for all $c \in \cC$, then for any $c_1,c_2\in \cC$ with $c_1\leq 
	c_2$, $T_1c_1 \leq T_1c_2 \leq T_2c_2$. Iterating from any $c\in \cC$ and 
	using Theorem~\ref{t:opt}, we obtain 
	$c_1^* = \lim_{n \to \infty}(T_1)^nc \leq \lim_{n \to \infty}(T_2)^nc 
	= c_2^*$, which proves the claim once $T_1c \leq T_2c$ is established.
	
	To show that $T_1c \leq T_2c$ for any $c\in \cC$, take any $(w,z,\theta)\in \SS$ and define $\xi_j=(T_jc)(w,z,\theta)$. To show $\xi_1 \leq \xi_2$, suppose on the contrary that $\xi_1 > \xi_2$. Then $\xi_1>0$. Since $c$ is increasing in $w$, it follows from the definition of the time iteration operator in \eqref{eq:tio}, $Y_1 \leq Y_2$, and $u''<0$ that
	\begin{align*}
		u'(\xi_2)>u'(\xi_1)
		&=\max \{\EE_{z,\theta} \hat{\beta} \hat{R}
		    (u'\circ c)(\hat{R}(w - \xi_1) + \hat{Y}_1, \hat{Z}, \hat \theta), 
		u'(w) \} \\
		&\geq \max \{ \EE_{z,\theta} \hat{\beta} \hat{R} 
		    (u'\circ c)(\hat{R}(w - \xi_2) + \hat{Y}_2, \hat{Z}, \hat \theta), 
		u'(w) \} 
		\geq u'(\xi_2),
	\end{align*}
	which is a contradiction. This completes the proof.
\end{proof}

\begin{proof}[Proof of Proposition~\ref{pr:binding}]
	Recall that, for all $c \in \cC$, $\xi(w,z,\theta) := Tc(w,z,\theta)$ solves
	\begin{align}
		\label{eq:T_opr_general}
		&(u'\circ \xi)(w,z,\theta) \nonumber \\
		&= \min \left\{
		\max \left\{ 
		    \EE_{z,\theta} \hat{\beta} \hat{R} 
		    (u'\circ c) (\hat{R} [w - \xi(w,z,\theta)] + \hat{Y}, \hat{Z},\hat \theta), u'(w) 
		\right\}, u'(0)
		\right\}.
	\end{align}
	For each $(z,\theta) \in \ZZ\times \Theta$ and $c \in \cC$, define
	\begin{equation}\label{eq:a_bar}
		\bar{w}_c (z,\theta) := (u')^{-1} \left[ 
			\min\left\{
			    \EE_{z,\theta} \hat{\beta} \hat{R} 
			    (u'\circ c)(\hat{Y}, \hat{Z},\hat \theta) 
			    , u'(0) 
			\right\}
		\right]
	\end{equation}
	and $\bar{w}(z,\theta) := \bar{w}_{c^*} (z,\theta)$. 
	
	Let $w \leq \bar{w}_c (z,\theta)$. We claim that $\xi(w,z,\theta) = w$. Suppose to the contrary that $\xi(w,z,\theta) < w$. Then by Lemma~\ref{lm:basic}, \eqref{eq:T_opr_general}, and the monotonicity of $c\in \cC$, we have
	\begin{align*}
		u'(w) &< (u'\circ \xi)(w,z,\theta) 
		= \min \left\{ 
		    \EE_{z,\theta} \hat{\beta} \hat{R} 
		    (u'\circ c)(\hat{R} \left[w - \xi(w,z,\theta) \right] + \hat{Y}, \hat{Z},\hat \theta),   
		    u'(0)
		\right\}\\ 
		&\leq \min \left\{
		    \EE_{z,\theta} \hat{\beta} \hat{R} 
		    (u'\circ c)(\hat{Y}, \hat{Z}, \hat\theta),
		    u'(0)
		\right\}
		= u'[\bar{w}_c (z,\theta)].
	\end{align*}
	From this we get $w > \bar{w}_c (z,\theta)$, which is a contradiction. Hence, $\xi(w,z,\theta) = w$.
	
	On the other hand, if $\xi(w,z,\theta) = w$, then by \eqref{eq:T_opr_general}, we have
	\begin{align*}
		u'(w) &= (u'\circ \xi)(w,z,\theta) 
		= \min \left\{
		    \max\left\{
		        \EE_{z,\theta} \hat \beta \hat R 
		        (u'\circ c) (\hat Y, \hat Z,\hat \theta),
		        u'(w)
		    \right\}, u'(0)
		\right\}  \\
		&\geq \min \left\{
		    \EE_{z,\theta} \hat \beta \hat R 
		    (u'\circ c) (\hat Y, \hat Z,\hat \theta), u'(0)
		\right\} 
		= u'[\bar w_c(z,\theta)].
	\end{align*}
	Hence, $w \leq \bar{w}_c (z,\theta)$. The first claim is verified. The second claim follows immediately from the first claim and the fact that $c^*$ is the unique fixed point of $T$ in $\cC$.
\end{proof}

\begin{proof}[Proof of Proposition~\ref{pr:lb}]
	Let $\bar s$ satisfy the conditions of the proposition and define
	\begin{equation*}
		\cC_1 := \{c \in \cC: c(w,z,\theta) \geq (1-\bar s) w
		\quad \text{for all } (w,z,\theta) \in \SS\}.
	\end{equation*}
	Since $T$ is an eventual contraction mapping on $\cC$ by Theorem~\ref{t:opt}, it suffices to show that $\cC_1$ is a closed subset of $\cC$ and that $T:\cC_1 \to \cC_1$.
	
	To see that $\cC_1$ is closed, let $\{c_n\}$ in $\cC_1$ and $c\in \cC$ such that $\rho(c_n,c) \to 0$. We need to verify $c\in \cC_1$. This obviously holds since $c_n(w,z,\theta) \geq (1-\bar s) w$ for all $(w,z,\theta)\in \SS$ and $\rho(c_n,c) \to 0$ implies $c_n(w,z,\theta) \to c(w,z,\theta)$ for all $(w,z,\theta)\in \SS$.
	
	To see that $T:\cC_1 \to \cC_1$, fix $c\in\cC_1$. We have $Tc \in \cC$ since $T$ is a self-map on $\cC$. It remains to show that $\xi := Tc$ satisfies $\xi (w,z,\theta) \geq (1 - \bar s) w$ for all $(w,z,\theta) \in \SS$. Suppose $\xi(w,z,\theta) < (1 - \bar s) w$ for some $(w,z,\theta) \in \SS$. Then by Lemma~\ref{lm:basic}, we have
	\begin{align*}
		u'((1 - \bar s) w) 
		&< (u' \circ \xi)(w, z,\theta)  \\
		&= \min\left\{
		\EE_{z,\theta} \hat{\beta} \hat{R} 
		\left( u' \circ c \right) 
		(\hat{R} \left[w - \xi(w,z,\theta) \right] + \hat{Y}, 
		\hat{Z}, \hat \theta), u'(0)
		\right\}  \\
		&\leq \min\left\{
		    \EE_{z,\theta} \hat{\beta} \hat{R}    
		    u'((1 - \bar s) \hat{R} \left[w - \xi(w,z,\theta) \right] 
		    + (1 - \bar s) \hat{Y}), u'(0)
		\right\}  \\
		&\leq \EE_{z,\theta} \hat{\beta} \hat{R} 
		u' [(1 - \bar s) \hat{R} \bar s w + (1 - \bar s) \hat{Y} ]    
		\leq \EE_{z,\theta} \hat{\beta} \hat{R} 
		u'[\bar s \hat R (1 - \bar s) w],
	\end{align*}
	which contradicts \eqref{eq:sbar} since $(1 - \bar s) w > 0$ and $(z,\theta) \in \ZZ\times \TT$. As a result, $\xi(w,z,\theta) \geq (1 - \bar s) w$ for all $(w,z,\theta) \in \SS$ and we conclude that $Tc \in \cC_1$.
\end{proof}

\begin{proof}[Proof of Proposition~\ref{pr:concavity}]
	Recall $\cC_0$ defined in \eqref{eq:cC_0}. Define
	\begin{equation*}
			\cC_2 := \left\{ c \in \cC_0 \colon 
			    w \mapsto c(w,z,\theta) \text{ is concave for all }
			    (z,\theta) \in \ZZ \times \Theta
			\right\}.
		\end{equation*}
	Since $T$ is an eventual contraction on $\cC$ by Theorem~\ref{t:opt}, and $\cC_0$ is a closed subset of $\cC$ and $T:\cC_0 \to \cC_0$ as was shown in the proof of Proposition~\ref{pr:monotonea}, to show that $c^*(w,z,\theta)$ is concave in $w$, it suffices to verify that $\cC_2$ is a closed subset of $\cC_0$ and $T:\cC_2 \to \cC_2$. The second claim follows immediately once the first claim is verified, since the concavity of $c^*$ in $w$ implies that $w \mapsto c^*(w,z,\theta)/w$ is a decreasing function, which has a natural lower bound zero.
	
	Clearly, $\cC_2$ is a closed subset of $\cC_0$ since limits of concave functions are concave. To see that $T:\cC_2 \to \cC_2$, fix $c \in \cC_2$. Since $T c \in \cC_0$, it remains to show that $w \mapsto \xi(w, z,\theta) := Tc (w,z,\theta)$ is concave for all $(z,\theta) \in \ZZ \times \Theta$. Given $(z,\theta)$, by Proposition~\ref{pr:binding}, we have $\xi(w,z,\theta) = w$ for $w \leq \bar{w}_c(z,\theta)$ and that $\xi (w,z,\theta) < w$ for $w > \bar{w}_c(z,\theta)$.  Since in addition $w \mapsto \xi(w,z,\theta)$ is continuous and increasing, to show the concavity of $\xi$ with respect to $w$, it suffices to show that $w\mapsto \xi (w,z,\theta)$ is concave on $(\bar{w}_c(z,\theta), \infty)$.
	
	Suppose there exist some $(z,\theta) \in \ZZ \times \Theta$, $\alpha \in [0,1]$, and $w_1, w_2 \in (\bar{w}_c (z,\theta), \infty)$ such that, for $w:=(1-\alpha) w_1 + \alpha w_2$, we have
	\begin{equation}
			\label{eq:as_ctdt}
			\xi \left(w, z,\theta \right)
			< (1 - \alpha) \xi(w_1, z,\theta) + \alpha \xi(w_2, z,\theta).
		\end{equation}
	With slight abuse of notation, let $\hat w := \hat{R} \left[w - \xi(w, z,\theta) \right] + \hat{Y}$. Applying \eqref{eq:as_ctdt} gives 
	\begin{equation*}
			\hat w \geq (\1-\alpha) \hat w_1 + \alpha \hat w_2.
		\end{equation*}
	Then by Proposition~\ref{pr:binding} and noting that consumption is interior, we obtain 
	\begin{align*}
			(u' \circ \xi) \left(w, z,\theta \right)  
			&= \EE_{z,\theta} \hat{\beta} \hat{R} 
			    (u' \circ c) (\hat w, \hat{Z}, \hat\theta)    
			\leq \EE_{z,\theta} \hat{\beta} \hat{R} (u' \circ c) 
			\left[
			    (1-\alpha) \hat w_1 + \alpha \hat w_2, \hat{Z}, \hat\theta
			\right].
		\end{align*}
	Substituting the CRRA utility function and using the concavity of $c$ yield
	\begin{align*}
		\xi \left(w, z,\theta \right)
		&\geq (u')^{-1} \left\{
		   \EE_{z,\theta} \hat{\beta} \hat{R} (u' \circ c) 
		   \left[
		       (1-\alpha) \hat w_1 + \alpha \hat w_2, \hat{Z}, \hat\theta
		   \right]
		\right\}  \\
		&= \left\{ 
			\EE_{z,\theta} \hat{\beta} \hat{R} 
			c \left[
				(1 - \alpha) \hat w_1 + \alpha \hat w_2, 
				\hat{Z}, \hat\theta 
			\right]^{-\gamma} 
		\right\}^{-1/\gamma}  \\
		&\geq \left\{ 
			\EE_{z,\theta} \hat{\beta} \hat{R} 
			\left[
				(1 - \alpha) c(\hat w_1,\hat Z, \hat\theta) 
				+ \alpha c (\hat w_2,\hat Z, \hat\theta)
			\right]^{-\gamma} 
		\right\}^{-1/\gamma} \\
		&= \left\{ 
			\EE_{z,\theta} 
			\left[
				(1 - \alpha) h(\hat w_1, \hat \nu) 
				+ \alpha h(\hat w_2, \hat \nu) 
			\right]^{-\gamma} 
		\right\}^{-1/\gamma},
	\end{align*}
	where for $\hat w \in \{\hat w_1, \hat w_2\}$ and $\hat \nu:=(\hat Z,\hat \theta, \hat\beta, \hat R)$, we denote
	\begin{equation*}
		h(\hat w, \hat \nu) := (\hat \beta \hat R)^{-1/\gamma} 
		c(\hat w, \hat Z,\hat \theta).
	\end{equation*}
	Then $\EE_{z,\theta} \left[h(\hat w_i, \hat \nu) \right]^{-\gamma} = \xi(w_i,z,\theta)^{-1/\gamma}$ for $i \in \{1,2\}$ based on the Euler equation. Applying the generalized Minkowski inequality (see, e.g., Theorem~198 of \cite{hardy1952inequalities}, page~146) then gives
	\begin{align*}
		\xi (w, z,\theta) 
		&\geq \left\{ 
		   \EE_{z,\theta} \left[
						(1 - \alpha) h(\hat w_1, \hat \nu) 
					\right]^{-\gamma} 
				\right\}^{-1/\gamma} + 
				\left\{ 
		   \EE_{z,\theta} \left[
		       \alpha h(\hat w_2, \hat \nu) 
		   \right]^{-\gamma} 
		\right\}^{-1/\gamma} \\
		&= (1 - \alpha)  \left\{ 
			\EE_{z,\theta} \left[h(\hat w_1, \hat \nu) \right]^{-\gamma} 
		\right\}^{-1/\gamma} + 
		\alpha \left\{ 
			\EE_{z,\theta} \left[h(\hat w_2, \hat \nu) \right]^{-\gamma} 
		\right\}^{-1/\gamma} \\
		&= (1-\alpha) \xi(w_1,z,\theta) + \alpha \xi(w_2,z,\theta).
	\end{align*}
	which contradicts \eqref{eq:as_ctdt}. Hence, $ w\mapsto \xi(w,z,\theta)$ is concave and $T : \cC_2 \to \cC_2$. The proof is now complete.
\end{proof}

\bibliographystyle{ecta}
\bibliography{os}

\end{document}